\documentclass[american,acmlarge]{article}
\usepackage[T1]{fontenc}
\usepackage[latin9]{inputenc}
\usepackage[a4paper]{geometry}
\usepackage{float}
\usepackage{mathtools}
\usepackage{amsmath}
\usepackage{amsthm}
\usepackage{amssymb}
\usepackage{graphicx}

\makeatletter

\providecommand{\tabularnewline}{\\}

\theoremstyle{plain}
\newtheorem{thm}{\protect\theoremname}
  \theoremstyle{definition}
  \newtheorem{defn}[thm]{\protect\definitionname}
  \theoremstyle{plain}
  \newtheorem{prop}[thm]{\protect\propositionname}

\usepackage{hyphenat}
\usepackage{booktabs}
\usepackage{enumitem}
\usepackage{tabto}

\usepackage[algoruled,linesnumbered,algo2e]{algorithm2e}

\SetAlFnt{\small}
\SetAlCapFnt{\small}
\SetAlCapNameFnt{\small}
\SetAlCapHSkip{0pt}
\IncMargin{-\parindent}

\@ifundefined{showcaptionsetup}{}{%
 \PassOptionsToPackage{caption=false}{subfig}}
\usepackage{subfig}
\makeatother

\usepackage{babel}
  \providecommand{\definitionname}{Definition}
  \providecommand{\propositionname}{Proposition}
\providecommand{\theoremname}{Theorem}

\begin{document}

\title{Quality and Cost of Deterministic Network Calculus --\\
Design and Evaluation of an Accurate and Fast Analysis\date{}}

\author{Steffen Bondorf\,\thanks{This work is supported by a Carl Zeiss Foundation grant.}
, Paul Nikolaus, Jens B. Schmitt\\
Distributed Computer Systems (DISCO) Lab,\\
University of Kaiserslautern, Germany}

\maketitle
\begin{abstract}
Networks are integral parts of modern safety-critical systems and
certification demands the provision of guarantees for data transmissions.
Deterministic Network Calculus (DNC) can compute a worst-case bound
on a data flow's end-to-end delay. Accuracy of DNC results has been
improved steadily, resulting in two DNC branches: the classical algebraic
analysis and the more recent optimization-based analysis. The optimization-based
branch provides a theoretical solution for tight bounds. Its computational
cost grows, however, (possibly super-)exponentially with the network
size. Consequently, a heuristic optimization formulation trading accuracy
against computational costs was proposed. In this article, we challenge
optimization-based DNC with a new algebraic DNC algorithm. 

We show that: 
\begin{enumerate}
\item no current optimization formulation scales well with the network size
and {\small \par}
\item algebraic DNC can be considerably improved in both aspects, accuracy
and computational cost.{\small \par}
\end{enumerate}
To that end, we contribute a novel DNC algorithm that transfers the
optimization's search for best attainable delay bounds to algebraic
DNC. It achieves a high degree of accuracy and our novel efficiency
improvements reduce the cost of the analysis dramatically. In extensive
numerical experiments, we observe that our delay bounds deviate from
the optimization-based ones by only 1.142\% on average while computation
times simultaneously decrease by several orders of magnitude.
\end{abstract}

\section{Introduction}

\label{sec:Introduction}Accurately bounding timing constraints is
a fundamental problem of systems analysis. Applied in the design phase
of communication networks, it allows for their certification while
preventing over-provisioning of resources. For networks, the main
resource is the forwarding capability of links and the crucial metric
to consider is the end-to-end delay of data flows. An example of
a networked system requiring certification are the Ethernet-based
Avionics Full-Duplex (AFDX) data networks embedded in modern Airbus
aircraft. These have to be verified against strict deadlines in order
to attain the necessary certification. Given these demands and a
formal worst-case model of the network, Deterministic Network Calculus
(DNC) can provide worst-case bounds on the communication delay in
general feed-forward networks. Indeed, DNC was used for the certification
of the AFDX backbone as found in the Airbus~A380~\cite{Grieu_PhDthesis,Geyer_CommMag}.
Other recent example applications for DNC can be found in shared networked
storage in order to meet tail latency QoS~\cite{PriorityMeister}
and in multi-tenant data centers~\cite{Jang:2015:SPM:2785956.2787479}.
\begin{figure}
\begin{centering}
\includegraphics[width=0.55\columnwidth]{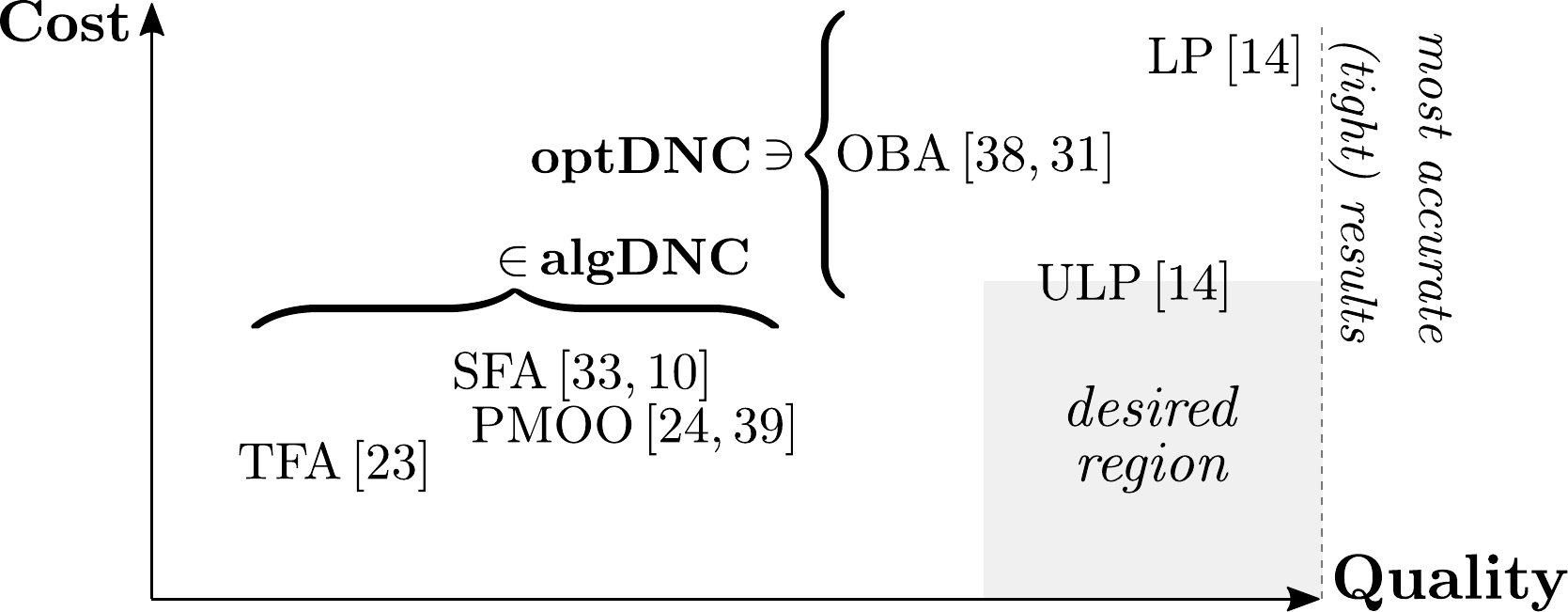}
\par\end{centering}
\centering{}\caption{\textbf{\label{fig:DNC-QandC-before}}Assumed quality and cost of
Deterministic Network Calculus (DNC) analyses, based on insufficient
evaluation data. While all analyses from the algebraic branch (algDNC)
are ``believed'' to be of insufficient accuracy, the optimization's
(optDNC's) ULP is currently regarded as the only analysis of good
quality as well as just feasible to execute.}
\end{figure}

DNC network analysis already allowed to derive end-to-end delay bounds
at an early stage~\cite{Cruz_ptTwo}, yet, achieving accurate results
has turned out to be a hard problem. DNC's evolution in the pursuit
of ever more accurate delay bounds has led to two branches: algebraic~DNC~(algDNC),
as used in the verification of AFDX, followed by optimization-based~DNC
(optDNC). Both share the same network model, yet, vastly differ in
the tools they apply to derive delay bounds, i.e., the actual network
analysis. AlgDNC uses its composition rules for operators to derive
an equation bounding the delay. OptDNC, on the other hand, derives
a linear program formulation whose solution bounds the end-to-end
delay.

OptDNC, the latest evolutionary step, is the more accurate branch
of DNC. In theory, it can compute most accurate, i.e., tight, bounds,
yet the effort to execute the according LP analysis grows \mbox{(possibly
super-)}exponentially with the network size. In fact, it becomes
computationally infeasible to analyze networks of even rather small
size and its authors accompanied their exact theoretical solution
with a practical variant, an optDNC heuristic called ULP~\cite{BouillardINFOCOM2010}. 

Unfortunately, the current plethora of DNC analyses, most of them
heuristics, was never evaluated comprehensively. Questions regarding
their fundamental tradeoff between quality (accuracy) and cost remain
open. Especially w.r.t. computational cost, there are crucial knowledge
gaps. Figure~\ref{fig:DNC-QandC-before} sketches the currently ``believed''
quality and cost relations of DNC analyses. These are vague estimates
based on very scarce evaluation data. We set out to fill that gap.
This allows us to make the following contributions:

\begin{enumerate}

\item Current optDNC design does not provide a fast heuristic for attaining delay bounds in feed-forward networks.

\end{enumerate}Both optDNC analyses of~\cite{BouillardINFOCOM2010},
LP and ULP, follow the same non-compositional design principle. The
given heuristic provides the only known way to trade accuracy against
computational efficiency and was believed to be just in the desired
region of feasibility and high quality (see Figure~\ref{fig:DNC-QandC-before}).
In this article, we show that this heuristic is actually very costly
and only applicable to small networks~\textendash{} though, it constitutes
the maximum reduction of effort by reducing the optimizations to
a single (but huge) linear program. The design principle thus inherently
limits the efficiency of the analysis.

\begin{enumerate}[resume]

\item Identification of current algDNC analyses' problems. Their impact needs to be minimized to improve algDNC accuracy.

\end{enumerate}The DNC branch of choice in the industrial context,
especially avionics, is the inaccurate algDNC~\cite{Geyer_CommMag,Optimize-AFDX-NC-ERTS-06,4638728,AFDX-Offset,Pegase_ERTSS}.
Therefore, we derive in-depth knowledge of its problems, the impact
of which needs to be minimized. We turn this knowledge into an analysis
design that achieves our objective.

The algDNC is computationally attractive due to its compositional
approach, yet, its inaccuracy was the very reason for DNC's branching
into optimization-based analyses. Regarding quality, algDNC, represented
by the recent SFA of~\cite{BouillardHDR}, cannot compete with optDNC's
ULP heuristic. Figure~\ref{fig:Comp-SFA-ULP-20DevExample} depicts
flow delay bounds in a small network of $20$ devices. Already in
this small network, the algebraically derived SFA bounds oscillate
wildly with a large amplitude above the ULP results. This behavior
can be observed in larger networks as well. We have taken statistics
over $9$ different topology sizes, amounting to $12376$ individual
flow delay bounds (cf.~Table~\ref{tab:NetServerFlow}, left, in
Appendix~\ref{subsec:SetNetworks}). This first comprehensive comparison
of an algDNC network analysis and an optDNC one reveals the gap in
accuracy we set to overcome in this article. SFA delay bounds deviate
from the respective ULP bounds by an average of $\approx\negmedspace15.2$\%,
the maximum is even as large as \textbf{$72.65$}\%. Figure~\ref{fig:Comp-SFA-ULP-DeviationDelay}
shows a more fine-grained distribution of deviations.

Hence, DNC currently offers an accurate, yet, computationally infeasible
branch and a reasonably fast, but inaccurate branch of network analyses.
\begin{figure}
\begin{centering}
\subfloat[\textbf{\label{fig:Comp-SFA-ULP-20DevExample}}Sample network with
20 devices.]{\begin{centering}
\includegraphics[width=0.45\textwidth]{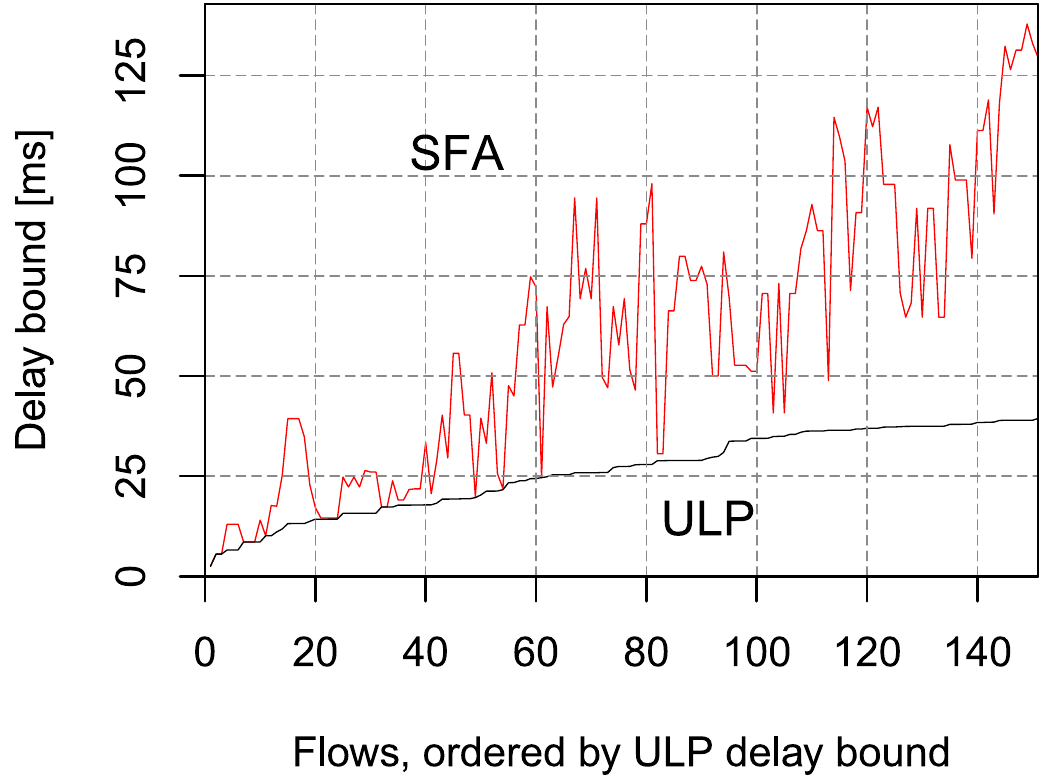}
\par\end{centering}
}\hspace*{7.5mm}\subfloat[\label{fig:Comp-SFA-ULP-DeviationDelay}Delay bound deviations across
$12376$ delay bounds.]{\begin{centering}
\includegraphics[width=0.45\textwidth]{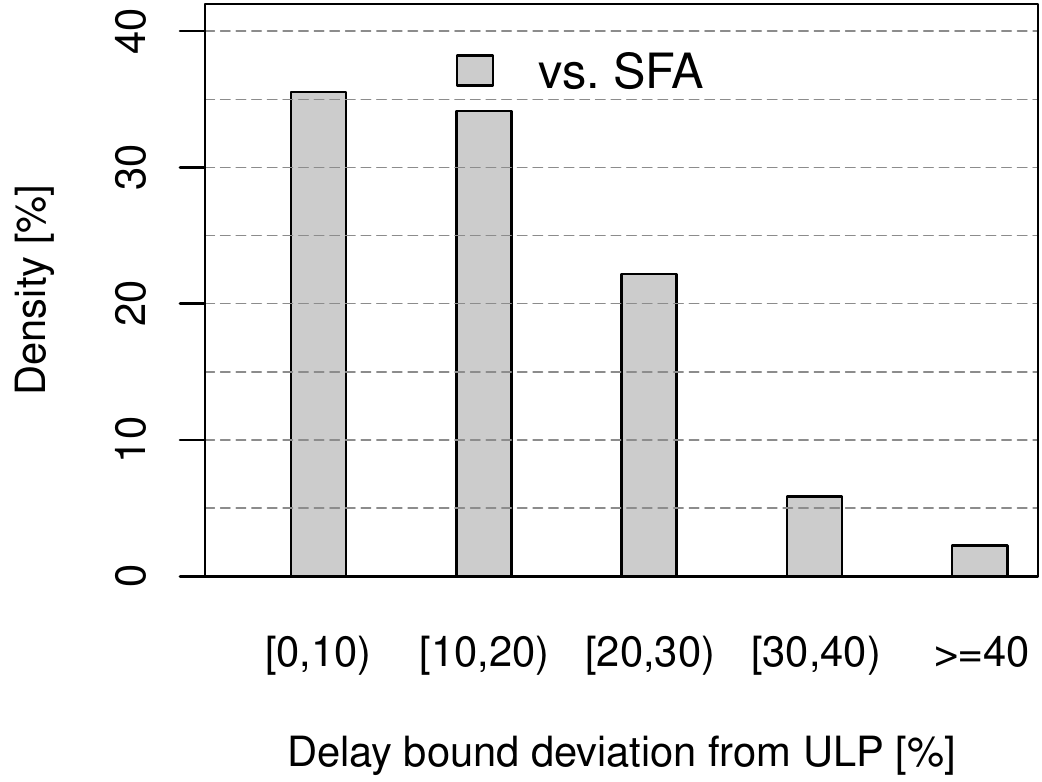}
\par\end{centering}
}
\par\end{centering}
\centering{}\caption{\label{fig:Delay-optDNC-algDNC-before}Delay bounds, algDNC (SFA)
vs. optDNC (ULP).}
\end{figure}

\begin{enumerate}[resume]

\item Design of a novel algDNC analysis algorithm, the first DNC heuristic that is accurate and fast at the same time.

\end{enumerate}Our new algDNC analysis achieves its accuracy by
incorporating concepts from optDNC's analysis design. Consequently,
it becomes computationally infeasible at first sight. Yet, we provide
an algorithm that vastly reduces the computational cost while accuracy
of delay bounds is unchanged. Both is shown by extensive benchmarks
where our new analysis achieves delay bounds that deviate from the
ULP by an average of only $1.142$\%. Moreover, our algorithm bounds
delays several orders of magnitude faster than the ULP.\medskip{}

The remainder of this article is structured as follows: Section~\ref{sec:NC-Background}
describes the modeling background of DNC. In Section~\ref{sec:OptimizationNC},
we evaluate the cost of optDNC and in Section~\ref{sec:Analysis_Framework}
we derive the problems causing loss of accuracy in algDNC, leading
to an analysis procedure that minimizes their impact. In Section~\ref{sec:algDNCwithOpt}
we design a novel algorithm for this analysis that is accurate and
fast at the same time. Section~\ref{sec:Evaluation} presents a comprehensive
numerical evaluation of all presented algDNC and optDNC analyses,
ranking them regarding their quality and cost. Section~\ref{sec:RelatedWork}
relates our contribution to the literature and Section~\ref{sec:Conclusion}
concludes the article.

\section{The Common Network Calculus Model}

\label{sec:NC-Background}

\subsection{Data Arrivals and Forwarding Service}

\label{subsec:System-Description}In DNC, flows are characterized
by functions that cumulatively count their data~\cite{Cruz_ptOne}.
These functions are non-negative, wide-sense increasing and pass through
the origin:
\begin{align*}
\mathcal{F}_{0}\,=\,\left\{ \left.f:\mathbb{R}\rightarrow\mathbb{R}_{\infty}^{+}\;\right|\;f\left(0\right)=0,\;\forall\tau\le t\,:\,f(\tau)\!\leq\!f(t)\right\} , & \;\mathbb{R}_{\infty}^{+}\coloneqq\text{\ensuremath{\left[0,+\infty\right)\cup\left\{  +\infty\right\} } }.
\end{align*}

A flow's functions $A(t)$ and $A'(t)$ describe its data put into
a server $s$ and its data put out from $s$, both from the start
of arrivals up until time $t$. Given these functions, the expected
delay of a data unit arriving at $s$ can be computed.
\begin{defn}
\label{def:Backlog-Delay-Generic}(Delay) Assume a flow with input
function $A$ traverses a server $s$ and results in the output function
$A'$. The \emph{(virtual)} delay for a data unit arriving at $s$
at time $t$ is
\[
D(t)=\inf\left\{ \tau\geq0\;\left|\;A(t)\leq A'(t+\tau)\right.\right\} \!.
\]
\end{defn}

Note, that this definition requires the server to preserve the order
of the flow's arrivals when forwarding its data.

In a second modeling step, network calculus introduces bounding functions
for data arrivals, so-called arrival curves.
\begin{defn}
\label{def:Arrival-Curve}(Arrival Curve) Let a flow have input function
$A\in\mathcal{F}_{0}$, then $\alpha\in\mathcal{F}_{0}$ is an arrival
curve for $A$ iff it bounds $A$ in any observation window of duration
$d$, i.e.,
\[
\forall t\,\forall d,\,0\leq d\le t\,:\,A(t)-A(t-d)\leq\alpha(d).
\]
\end{defn}

Using arrival curves, only the duration of observation is required
to obtain an upper bound on the cumulative arrivals of data. I.e.,
flow arrivals' absolute starting point in time as well as the history
up until time $t$ need not be known anymore.

A useful basic shape for an arrival curve is the so-called token bucket.
It is enforced by the eponymous traffic regulation algorithm. These
curves are from the set $\mathcal{F}_{\text{TB}}\subseteq\mathcal{F}_{0}$,
\[
\mathcal{F}_{\text{TB}}\!=\!\left\{ \!\left.\gamma_{r,b}\,\right|\,\gamma_{r,b}(d)=\begin{cases}
0 & \text{if }d=0\\
b+r\cdot d & \text{otherwise}
\end{cases},\,r,b\in\mathbb{R}_{\infty}^{+}\right\} \!,
\]
where $r$ denotes the maximum arrival rate and $b$ is the maximum
burstiness (bucket size).

Scheduling and buffering at a server result in the output function~$A'(t)$.
Network calculus captures the minimum forwarding capabilities that
lead to $A'$ in interval time as well:
\begin{defn}
\label{def:Service-Curve}(Service Curve) If the service provided
by a server~$s$ for a given input~$A$ results in an output~$A'$,
then~$s$ is said to offer a service curve $\beta\in\mathcal{F}_{0}$
iff
\[
\forall t\,:\,A'(t)\,\geq\,{\displaystyle \inf_{0\leq d\leq t}}\left\{ A(t-d)+\beta(d)\right\} \!.
\]
\end{defn}

A number of servers fulfill a stricter definition of service curves
by considering their state in addition to input~$A$.
\begin{defn}
\label{def:BackloggedPeriod}(Backlogged Period) A server~$s$ with
input~$A$ and output~$A'$ is backlogged during period $(\underline{t},\overline{t})$,
if $\forall t\in(\underline{t},\overline{t})\,:\,A(t)>A'(t)$.
\end{defn}

Servers offering strict service guarantees have a higher output guarantee
during any backlogged period.
\begin{defn}
\label{def:Strict-Service-Curve}(Strict Service Curve) If, during
any backlogged period of duration $d$, a server~$s$ with input~$A$
and output~$A'$ guarantees an output of at least $\beta(d)$, it
offers a \emph{strict} service curve $\beta\in\mathcal{F}_{0}$.
\end{defn}

A basic shape for service curves is the rate latency function defined
by the set $\mathcal{F}_{\text{RL}}\subseteq\mathcal{F}_{0}$,
\[
\mathcal{F}_{\text{RL}}\!=\!\left\{ \beta_{R,T}\,\left|\,\beta_{R,T}(d)\!=\!\max\!\left\{ 0,\,R\!\cdot\!(d-T)\right\} \!,\;R,T\!\in\!\mathbb{R}_{\infty}^{+}\right.\!\right\} \!,
\]
where $R$ denotes the minimum service rate and $T$ is the maximum
latency. Networks where curves are exclusively from the sets $\mathcal{F}_{\text{TB}}$
and $\mathcal{F}_{\text{RL}}$ have been in the focus of recent advances
in DNC~\cite{Bouillard2015270,BS15}. 

Given these bounding curves, the maximum delay experienced by a flow
when crossing a server is simply computed by the horizontal deviation
between arrival curve~$\alpha$ and service curve~$\beta$: 
\[
\sup_{d\geq0}\left\{ \inf\left\{ \tau\geq0\,:\,\alpha(d)-\beta(d+\tau)\leq0\right\} \right\} .
\]

\subsection{The Network Model}

\label{subsub:Server-Graph}

\subsubsection{Arbitrary Multiplexing}

Consistent with our worst-case perspective, we make no assumption
on the ordering between flows when they are multiplexed at servers.
That is we employ the so-called arbitrary multiplexing when computing
left-over service for a flow of interest at a shared server. Certainly,
if we could, for instance, assume FIFO multiplexing, then the performance
bounds could be improved. However, in applications requiring worst-case
guarantees such as, e.g., for the certification in avionics, a FIFO
assumption may already be considered too optimistic. In fact, many
network switches use highly optimized internal switching fabrics that
can lead to reordering in order to avoid head-of-line blocking~\cite{karol1987input}.
Overall, by using arbitrary multiplexing we play safe with respect
to the worst-case.

\subsubsection{The Feed-forward Property}

A second assumption of current DNC analyses for networks is the absence
of cyclic dependencies between flows. Work departing from this assumption
can be found in~\cite{Schioler_NCcyclDependencies,Thiele_CyclDependPerfAnalysis}
but is not covered by this article. Therefore, we focus on networks
that guarantee the feed-forward property by design. I.e., the analyzed
network does not allow for cycles. In Appendix~\ref{sec:GenerateTestNetworks},
we present a means to generate feed-forward networks for evaluation
of DNC analyses.

\section{Optimization-based DNC}

\label{sec:OptimizationNC}In this section, we perform a first cost
evaluation of optDNC in larger feed-forward networks. Bounding worst-case
delays tightly is a NP-hard problem~\cite{BouillardINFOCOM2010}.
Therefore, we examine optDNC's LP analysis that derives these tight
bounds and provide deeper insights into the computational cost of
optDNC's less accurate heuristic, the ULP. We compare our results
to the literature's algebraic analysis SFA, as used in recent work
of~\cite{BouillardINFOCOM2010,BouillardHDR}. The comparison reveals
that that DNC currently requires to choose between a computationally
barely feasible, yet accurate optimization analysis and a feasible,
but inaccurate algebraic analysis.

\subsection{The Tight LP Analysis}

\label{subsec:LPcombinatorialExplosion}OptDNC's LP analysis~\cite{BouillardINFOCOM2010}
takes the following steps to derive a given flow of interest's (foi's)
tight delay bound:
\begin{enumerate}
\item Starting at the foi's sink server, \emph{backtrack} all flows in order
to derive the dependencies between starts of backlogged periods of
servers. For simple tandems of servers, this step results in a total
order. However, in more general feed-forward networks, it derives
a partial order.
\item The partial order is extended to the set of all \emph{compatible total
orders}. This step enumerates all relations~of servers' backlogged
period beginnings. It is subject to restrictions caused by flows that
rejoin again after demultiplexing.
\item Each total order is converted to one linear program that also includes
the network description's constraints such as arrival curves and service
curves~(Section~\ref{subsec:System-Description}). The \emph{maximum
of all their solutions} constitutes the tight delay.
\end{enumerate}
The third step shows that the LP is an all-or-nothing analysis where
we cannot judge validity of a delay bound before computing all of
them. Unfortunately, step~2 is prone to a combinatorial explosion
that constitutes the underlying reason for this DNC analysis' \mbox{(possibly
super-)}exponential growth in effort. To illustrate this problem,
let us briefly discuss on the number of linear programs (LPs) that
have to be solved in a sink-tree network: In the best case, we have
a pure tandem network of $n$ servers and then a single LP results;
in the worst case, we have a so-called fat tree with one root node
and $n-1$ leaf nodes directly connected to it, resulting in $(n-1)!$
LPs. In a full binary tree, the number of LPs is lower bounded by
$\Omega\left(\left(\frac{n}{2}\right)!\right)$ \cite{Rus03}. In
general, calculating the number of total orders being compatible with
a given partial order is itself not a simple problem. One solution
is the Varol-Rotem algorithm~\cite{VR81}; we implemented this algorithm
to provide some numbers for the case of full $k$-ary trees (Table~\ref{tab:LinearExtensionsTrees}).
It is obvious that the computational cost to solve such large numbers
of linear programs becomes prohibitive quickly, even for moderately
sized networks.
\begin{table}[H]
\centering{}\caption{\label{tab:LinearExtensionsTrees}Number of LPs to solve for full
$k$-ary trees of moderate size.}
\begin{tabular}{rr|c|c|c|c|c}
 & \multicolumn{1}{r}{} & \multicolumn{4}{c}{Height $h$} & \tabularnewline
 &  & \multicolumn{1}{c}{0} & \multicolumn{1}{c}{1} & \multicolumn{1}{c}{2} & 3 & \tabularnewline
\cline{2-6} 
 & 1 & 1 & 1 & $1$ & $1$ & \rule{0pt}{3.25mm}\tabularnewline
\cline{3-6} 
Outdegree $k$ & 2 & 1 & 2 & $80$ & $21,964,800$ & \rule{0pt}{3.25mm}\tabularnewline
\cline{3-6} 
 & 3 & 1 & 6 & $7,484,400$ & $3.54\cdot10^{37}$ & \rule{0pt}{3.25mm}\tabularnewline
\cline{3-6} 
 & 4 & 1 & 24 & $3.89\cdot10^{15}$ & $1.12\cdot10^{110}$ & \rule{0pt}{3.25mm}\tabularnewline
\cline{2-6} 
\end{tabular}
\end{table}

\subsection{The Accurate, but Costly ULP Heuristic}

\label{subsec:optDNC} Based on the LP's optimization formulation
and the analysis design it was accompanied with, a heuristic called
Unique LP (ULP) was proposed \cite{BouillardINFOCOM2010}. It circumvents
the combinatorial explosion by skipping step~2 from above. I.e.,
it always derives a single linear program that is directly based on
the partial order from step~1. Thus, the ULP does not attain tight
delays but it was shown to stay very close to these in a (small) sample
network. Its computational cost in terms of analysis run time~\cite{FinitaryRTC,LBS16}
in larger, more general feed-forward networks had not been investigated
yet.

\subsubsection{Inaccuracy due to Paying Segregation more than Once}

\label{subsec:PSOOvilation_generic}The ULP models multiplexing of
flows with joint backlogged periods and has access to global knowledge
in the optimization step. Yet, recent work~\cite{BS16-1} shows how
skipping step~2 reduces the accuracy of bounds attained by the ULP.
In the following situation, these two features of optDNC are only
utilized by the LP analysis: Assume flows multiplex at a server and
demultiplex after it, then they take different paths and multiplex
later at a different server again. In such a network, the analysis
should strive for the Pay Segregation Only Once (PSOO) principle~\textendash{}
separated paths should not result in segregated flow analysis and
overly pessimistic service allocation to flows. Yet, the ULP cannot
capture the (global) interdependency between the former and the latter
server where the flows jointly cause a backlogged period. Instead,
it indeed derives an overly-pessimistic worst-case result by violating
the PSOO principle. This PSOO violation, however, occurs in all known
DNC heuristics, optimization and algebraic analysis.

\subsubsection{Computational Cost Evaluation of the ULP}

We implemented the ULP to measure its analysis execution time~\textendash{}
deriving the partial order, the linear program and solving it~\textendash{}
for our set of feed-forward networks (Appendix~\ref{subsec:SetNetworks}).
For the optimization, we employ two different solvers: the open-source
LpSolve~5.5.2.0 and IBM~CPLEX 12.6.2. All computations were executed
on a physical machine equipped with two Intel Xeon E5420 server CPUs
(4~physical cores each) running with a clock speed of 2.5GHz and
a total of 12GB~RAM.
\begin{figure}
\begin{centering}
\subfloat[\label{fig:Comp-SFA-ULP-ExecTime}Scaling of network analysis time.]{\begin{centering}
\includegraphics[width=0.45\textwidth]{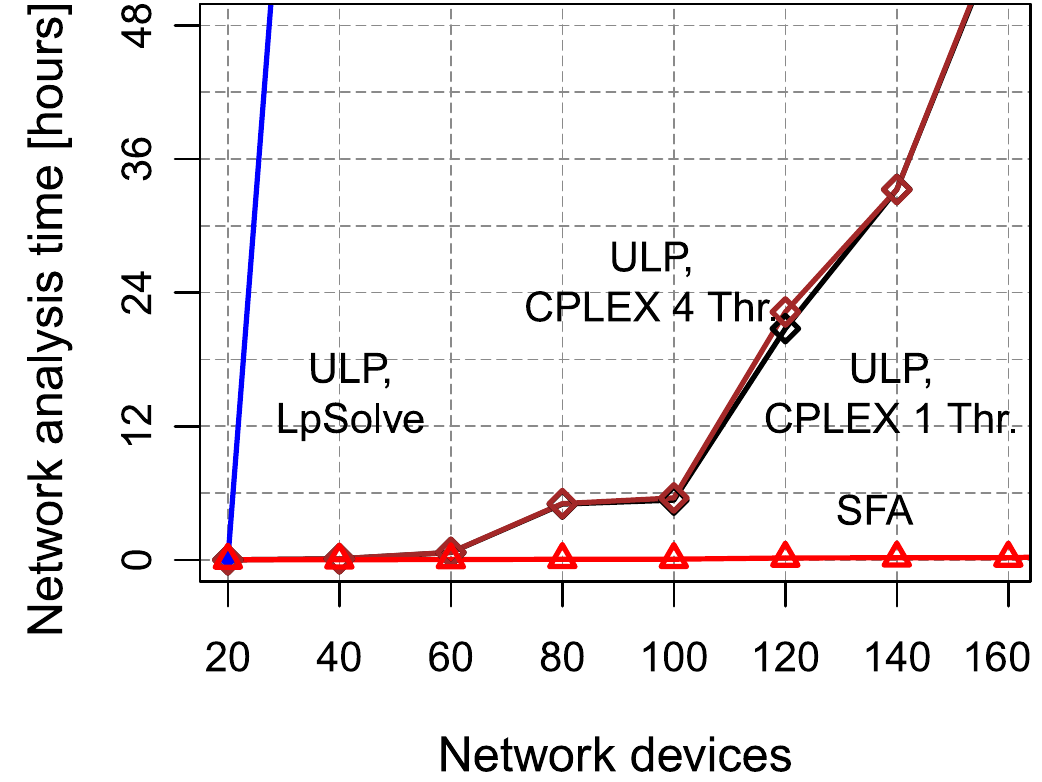}
\par\end{centering}
}\hspace*{7.5mm}\subfloat[\label{fig:ULP-Share-DiscoDNC-CPLEX}Share of CPLEX execution time.]{\begin{centering}
\includegraphics[width=0.45\textwidth]{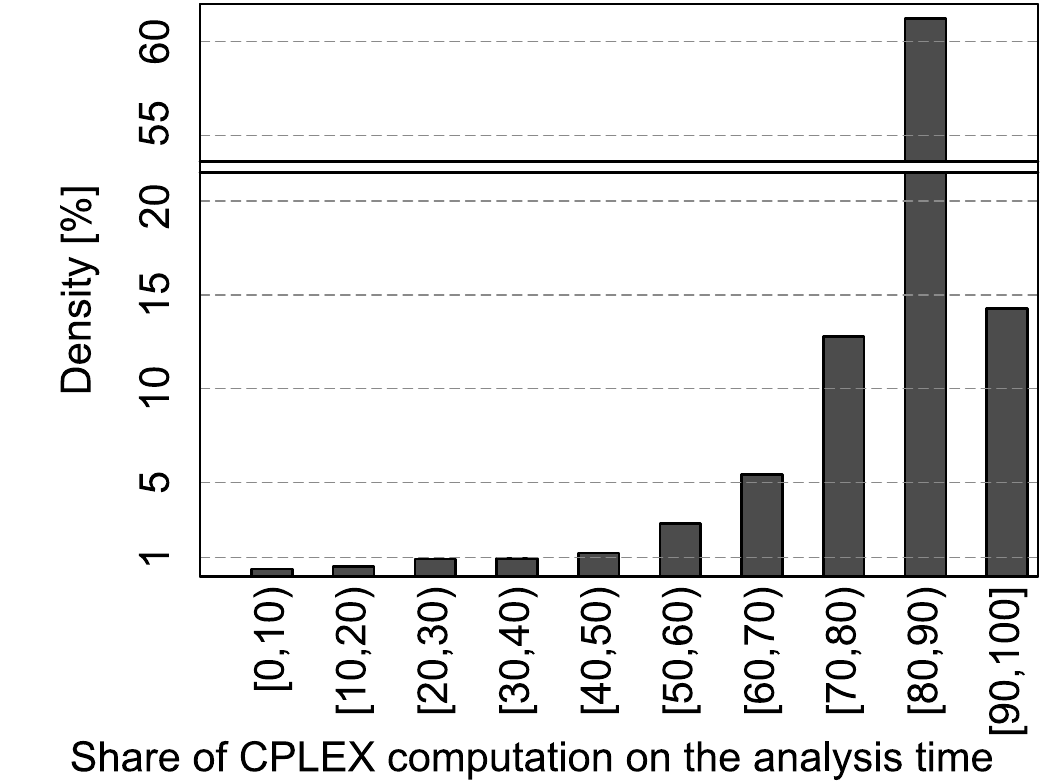}
\par\end{centering}
}
\par\end{centering}
\centering{}\caption{\label{fig:Comparison_SFA_ULP}Analysis cost (run time of a network
analysis).}
\end{figure}

Figure~\ref{fig:Comp-SFA-ULP-ExecTime} shows our results and compares
them to algDNC's SFA (see Section~\ref{subsec:Tandem-Analyses},~\cite{BouillardHDR}),
the contender used in~\cite{BouillardINFOCOM2010}. These results
defeat the hope in the ULP. The choice of tooling becomes crucial
quickly as LpSolve already struggles with $40$ devices, with regard
to the computation time as well as derivations of results. IBM~CPLEX
performs considerably better, yet, the ULP analysis time still increases
very fast with the network size. This trend is decisively negative
as, among the three steps, optimization with CPLEX takes on average
$>80$\% of the analysis time (cf. Figure~\ref{fig:ULP-Share-DiscoDNC-CPLEX}).
For instance, the ULP analysis already requires $13$ days at less
than \textbf{$200$} devices. Moreover, applying it in a design space
exploration for a network with as little as $80$ devices does not
seem feasible either as the analysis time of a single alternative
exceeds $5$ hours (cf. Figure~\ref{fig:Comp-SFA-ULP-ExecTime}).
The SFA, on the other hand, finishes analyzing these networks in just
a small fraction of these times.

\section{Quality of Algebraic DNC: Problems and Prospects}

\label{sec:Analysis_Framework}Algebraic DNC (algDNC) derives a (min-plus
algebraic) equation that computes the delay bound for the flow of
interest. We show that the algDNC analyses currently do not derive
the best equation possible and we provide the theoretical foundation
to solve this issue. In consequence, the best algDNC delay bound can
now be obtained.

\subsection{(min,+)-Algebraic Calculus}

\label{subsec:Network-Calculus-Analysis}First, we provide the basic
results we rely on in algDNC and thus in our accurate and fast analysis
design.

\subsubsection{Basic Operations}

AlgDNC's equations consist of (min,+)-operations \cite{LeBoudec_NCbook,Chang_NCbook}
that are computationally attractive. E.g., for the above curves in
$\mathcal{F}_{\text{TB}}$ and $\mathcal{F}_{\text{RL}}$, they can
be implemented in $\mathcal{O}(1)$~\cite{Bouillard_JDEDS}.
\begin{defn}[\emph{$(\min,+)$-Operations}]
\label{def:MinPlusOperations-1}\emph{The $(\min,+)$-}algebraic
operations for two network calculus curves $f,g\in\mathcal{F}_{0}$
are:
\begin{itemize}
\item \emph{Aggregation: $\left(f+g\right)\!(t)=f(t)+g(t)$}
\item \emph{Convolution: $\left(f\otimes g\right)\!(t)=\inf_{0\le u\le t}\left\{ f(t-u)+g(u)\right\} $}
\item \emph{Deconvolution: $\left(f\oslash g\right)(d)=\sup_{u\geq0}\left\{ f(d+u)-g(u)\right\} $}
\end{itemize}
\end{defn}

The service curve definition then translates to $A'\geq A\otimes\beta$,
the arrival curve definition to $A\otimes\alpha\ge A$, and performance
bounds can be derived using the deconvolution operation:
\begin{thm}[Performance Bounds]
\emph{\label{thm:Performance-Bounds}}Consider a server $s$ that
offers a service curve $\beta$. Assume flow $f$ with arrival curve
$\alpha$ traverses $s$. We get the following bounds for $f$:
\begin{itemize}
\item Delay: $\forall t\in\mathbb{R}^{+}\hspace{-0.3mm}:\;D\left(t\right)\leq\inf\left\{ d\geq0\,|\,\left(\alpha\oslash\beta\right)(-d)\leq0\right\} $
\item Output: $\forall d\in\mathbb{R}^{+}\hspace{-0.3mm}:\hspace{0.8mm}\alpha'\hspace{-0.3mm}\hspace{-0.4mm}\left(d\right)\!=\hspace{0.25mm}\left(\alpha\oslash\beta\right)\left(d\right)$,
where $\alpha'$ is an arrival curve\footnote{Note, that arrival curves need to pass through the origin. In a slight
abuse of notation, we use the symbol $\oslash$ for both, deconvolution
and output bounding.} for $A'$.
\end{itemize}
\end{thm}

\subsubsection{Tandem Analysis}

\label{subsec:Tandem-Analyses} AlgDNC focuses the delay analysis
on a specific flow of interest (foi), end-to-end on its path (a tandem
of servers).
\begin{thm}[Concatenation of Servers]
\emph{\label{thm:ConcatenationThmTandem}}Let a flow $f$ cross a
tandem of servers $\mathcal{T}=\left\langle s_{1},\ldots,s_{n}\right\rangle $
and assume that each $s_{i}$, $i\in\left\{ 1,\ldots,n\right\} $,
offers a service curve $\beta_{s_{i}}$. The overall service curve
offered to $f$ is their concatenation 
\begin{eqnarray*}
\beta_{\mathcal{T}} & = & \bigotimes_{i=1}^{n}\beta_{s_{i}}.
\end{eqnarray*}
\end{thm}

\begin{thm}[Left-Over Service Curve]
\label{thm:Single-System-Left-Over-Service-Curve}Consider a server
$s$ that offers a strict service curve $\beta_{s}$. Let $s$ be
crossed by two flow aggregates $\mathbb{F}_{0}$ and $\mathbb{F}_{1}$
with arrival curves $\alpha^{\mathbb{F}_{0}}$ and $\alpha^{\mathbb{F}_{1}}$,
respectively. Then $\mathbb{F}_{1}$'s worst-case resource share
under arbitrary multiplexing at $s$, i.e., its left-over service
curve, is\emph{
\[
\beta_{s}^{\text{l.o.}\mathbb{F}_{1}}=\beta_{s}\ominus\alpha^{\mathbb{F}_{0}}
\]
}where $\left(\beta\ominus\alpha\right)(d)=\sup_{0\le u\le d}\left(\beta-\alpha\right)(u)$
denotes the non-decreasing upper closure of $\left(\beta-\alpha\right)(d)$.
\end{thm}

In this article, we use an open-source implementation of these operations
\cite{BS14}. An algDNC analysis of the flow's tandem then defines
rules for their composition such that the model's worst case is retained.

\paragraph{Separate Flow Analysis (SFA)~\cite{LeBoudec_NCbook}}

The first tandem analysis of algDNC, Separate Flow Analysis (SFA),
constitutes a straight-forward application of Theorem~\ref{thm:Single-System-Left-Over-Service-Curve}
and Theorem~\ref{thm:ConcatenationThmTandem}: In order to derive
the foi's service curve on its path $\mathcal{T}=\left\langle \ldots\right\rangle $,
cross-traffic is subtracted server by server and the per-server $\beta^{\text{l.o.}}$
are (min-plus) convolved. The resulting $\beta_{\mathcal{T}}^{\text{l.o.}}$
is used to bound the foi's delay. This procedure has a disadvantage:
If cross-flows share longer \mbox{(sub-)}paths with the foi, then
multiplexing is considered more than once in the equation bounding
the foi's delay.

\paragraph{Pay Multiplexing Only Once (PMOO) Analysis~\cite{Fidler:2003:ENC:646464.691395,SZM08-1}}

PMOO is an analysis principle that tackles the SFA's problem. It suggests
to convolve the analyzed sub-tandems of servers as much as possible
before subtracting cross-traffic. We denote this computation by
\[
\beta_{\mathcal{T}}^{\text{l.o.}}=\beta_{\mathcal{T}}\ominus\left(\alpha^{\mathbb{F}_{i}},\ldots,\alpha^{\mathbb{F}_{i+j}}\right)
\]
where $\mathcal{T}$ is the analyzed tandem, $\beta_{\mathcal{T}}$
is its convolved service curve and $\left(\mathbb{F}_{i},\ldots,\mathbb{F}_{i+j}\right)$
denotes the cross-flow aggregates to subtract. For single-server tandems
$\left\langle s\right\rangle $, the derivation equals Theorem~\ref{thm:Single-System-Left-Over-Service-Curve},
i.e., $\beta_{\left\langle s\right\rangle }^{\text{l.o.}}=\beta_{\left\langle s\right\rangle }\ominus\left(\alpha^{\mathbb{F}_{i}},\ldots,\alpha^{\mathbb{F}_{i+j}}\right)=\beta_{s}\ominus\sum_{k=0}^{j}\alpha^{\mathbb{F}_{i+k}}$.
For larger tandems and more involved interference patterns, \cite{SZM08-1}~provides
the first theoretical implementation of this principle (code is in
\cite{BS14}). We use this PMOO analysis in our article. It results
in more accurate bounds than SFA in many scenarios, with some noticeable
exceptions which we detail later~\cite{Jens_ArbMuxLurch}.

\paragraph{Validity of PMOO Results}

The PMOO\textquoteright s underlying principle to convolve service
curves before subtracting the impact of cross-traffic arrivals is,
in fact, not covered by the basic rules to compose algDNC operations.
The left-over service curve operation can only be applied to \emph{strict}
service curves as only these guarantee a bound on the backlogged period.
In the worst case, this bound needs to be reached in order to assume
left-over service. Yet, the convolution of two strict service curves
does not necessarily result in a strict service curve. Correctness
of the PMOO result is guaranteed nonetheless. The algDNC literature
offers multiple justifications: On the one hand, it was proved that
the algebraic manipulations are inherently pessimistic, such that
the optimistically assumed bound on the backlogged period of simple
service curves is set off. This is shown by proving that (valid) results
of an optimization analysis are always better than PMOO results~\cite{Jens_ArbMuxLurch,KGS10-1}.
Hence, the PMOO result constitutes a valid bound as well. On the other
hand, more direct proofs arguing over input/output relations of flows
crossing a tandem of servers can be found in the literature, too~\cite{Bouillard_OptRoutingJournal,Bouillard_SCdo_dont}.

\subsection{Compositional Feed-forward Analysis: From Tandem Analyses to Network
Analyses}

\label{subsec:Feed-forward-Analyses}\label{subsec:alg_comp_DNC}\label{subsec:Problems-of-algDNC}The
algDNC tandem analyses are used in a compositional fashion to derive
delay bounds in feed-forward networks (compFFA). Conceptually, the
procedure can be split into two parts \cite{BS15}:
\begin{description}
\item [{i)}] The analysis abstracts from the feed-forward network to the
foi's path. In this part, arrivals of cross-traffic are bounded at
the locations of interference with the foi.
\item [{ii)}] A tandem analysis derives the foi's delay bound.
\end{description}
Whereas the tandem to analyze in the second part is known from the
start, the tandems in the first part need to be derived. I.e., the
network must to be decomposed into a sequence of tandems to analyze
(derive $\beta^{\text{l.o.}}$s). This decomposition depends on the
applied analysis. Tandems are then interfaced with the output bound
operation $\oslash$. We illustrate the decomposition in Figure~\ref{fig:Sample-server-topology},
a minimal network that requires both parts of the procedure. Given
the global view in this sample network, we see that SFA decomposes
each flow's paths into smallest possible tandems (i.e., individual
servers). In contrast, the PMOO analysis ``decomposes'' into
longest possible tandems to apply its left-over service curve derivation
to.

These decompositions define orders of algebraic operators in order
to attain worst-case results. However, current algDNC analyses do
not exploit all degrees of freedom to decompose a feed-forward network
into tandems. This can cause pessimism and reduce the accuracy of
results. In the following, we provide detailed insight into these
aspects. Last, note that the presented optDNC analyses are not compositional.
They directly derive the foi's delay bound.

\subsubsection*{Causes for Overly-pessimistic Results of compFFA Heuristics}

We provide an analysis of problems encountered in compFFA. There are
several situations that algebraic DNC's compositional feed-forward
analysis can only handle by overly-pessimistic worst-case approximations.
Detailed knowledge about them enables us to derive mitigation strategies
that reduce occurrences to a minimum. We present newly found phenomena
as well as known ones that lack comprehensive evaluation. All are
considered in the design of our novel algebraic analysis algorithm.
Our mitigation strategies exploit unused degrees of freedom in composing
algDNC equations. Thus, the derived worst-case bounds are less pessimistic
than the known ones but still more pessimistic than optDNC's ULP results.
In Section~\ref{sec:Evaluation}, we evaluate the respective accuracy
differences.
\begin{figure}
\begin{centering}
\subfloat[\label{fig:Sample-server-topology}Minimal network with both compFFA
parts.]{\begin{centering}
\includegraphics[width=0.45\textwidth]{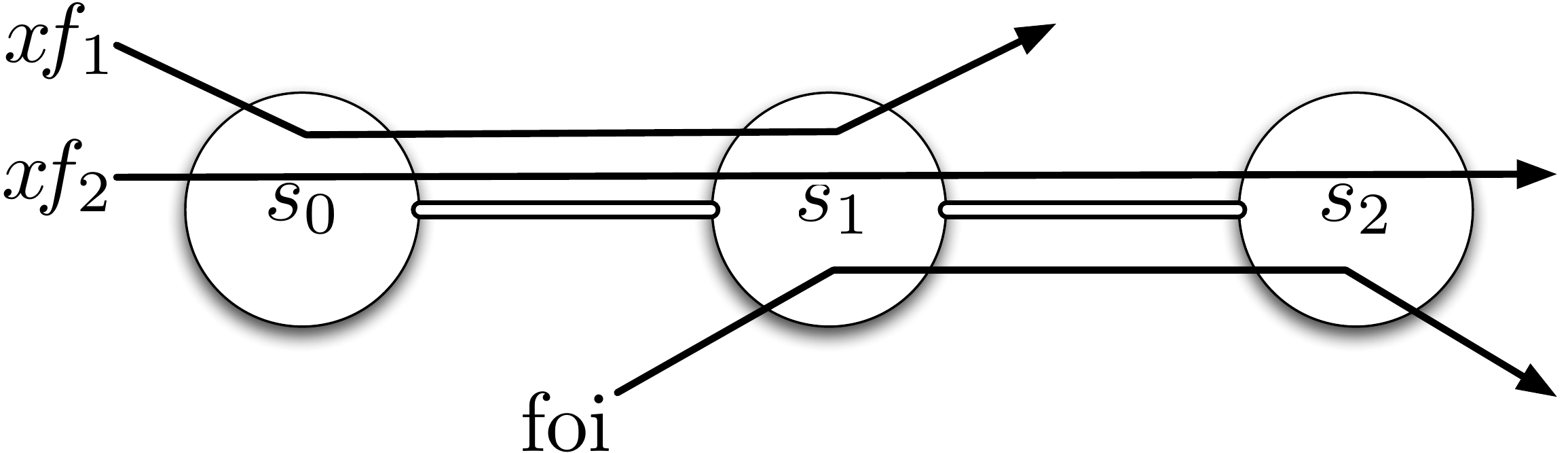}
\par\end{centering}
}
\par\end{centering}
\begin{centering}
\subfloat[\label{fig:BetaLo_PMOO}Construction of PMOO's (min,+)-equation.]{\begin{centering}
\includegraphics[width=0.45\textwidth]{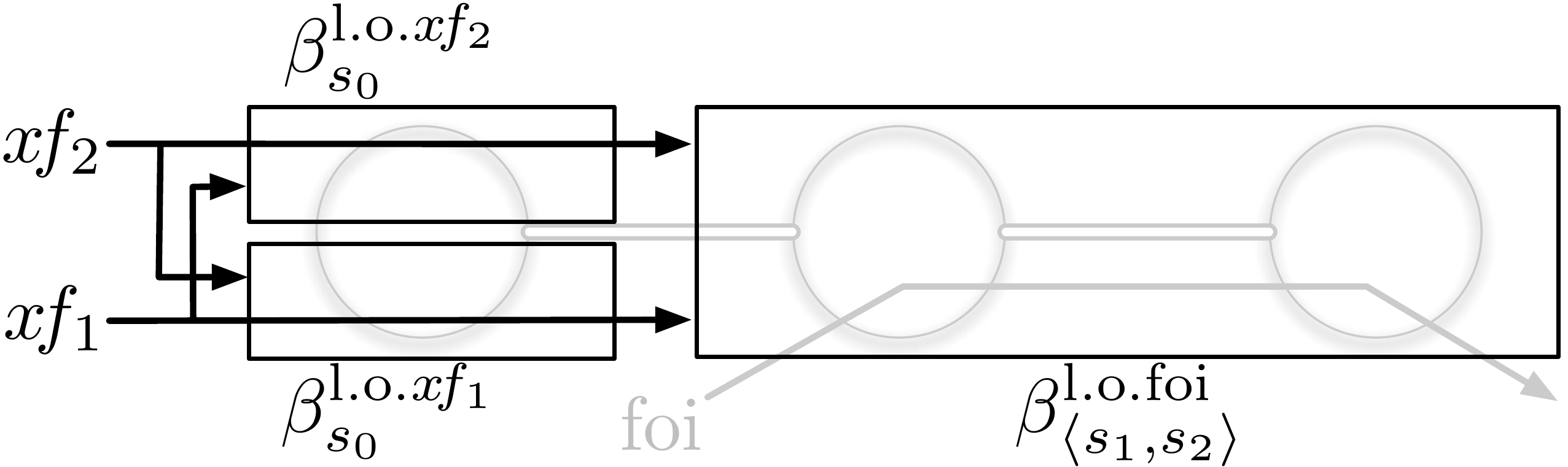}
\par\end{centering}
\centering{}}\hspace*{10mm}\subfloat[\label{fig:BetaLo_SFA}Construction of SFA's (min,+)-equation.]{\begin{centering}
\includegraphics[width=0.45\textwidth]{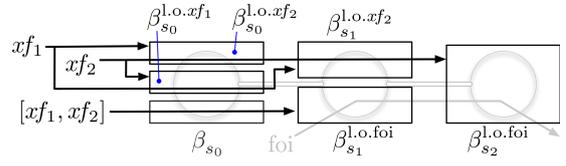}
\par\end{centering}
}
\par\end{centering}
\begin{centering}
\caption{Network with derivations of (min,+)-equations that illustrate the
restricting design of compFFA (Section~\ref{subsec:ProblemDesign-compFFA}).}
\vspace{1mm}
\par\end{centering}
\raggedright{}\emph{Interpretation:} Boxes depict tandems for $\beta^{\text{l.o.}}$
computation and arrows depict flows. Flows pointing at a box are subtracted
from a $\beta$, flows crossing a box use the respective $\beta^{\text{l.o.}}$
to bound their output from it.
\end{figure}

\subsubsection{Pay Segregation Only Once Violations in compFFA}

\label{subsec:EnforcedSegregations}Compositional feed-forward analysis
compFFA causes the same PSOO violations as presented in Section~\ref{subsec:PSOOvilation_generic}.
When backtracking flows, the left-over service curve operation computes
a worst-case service share for this flow. I.e., the result is computed
under local information only; interdependencies between flows that
multiplex at a server $s_{n}$, demultiplex after it, take different
paths and multiplex later at a different server $s_{m}$ again is
not considered. If two flows $f_{i}$ and $f_{j}$ are each backtracked
from $s_{m}$ to $s_{n}$, they compute their respective left-over
service curves $\beta_{s_{n}}\ominus f_{i}$ and $\beta_{s_{n}}\ominus f_{j}$
at server~$s_{n}$. Both flows simultaneously assume that the opposite
flow has a higher priority and is served first. This composition of
algebraic operations is permissible, meaning it represents a worst
case, leading to valid results. However, the total service assumed
to be available to both flows is less than the server's minimum guarantee
$\beta_{s_{n}}$. It is shown that PMOO and ULP suffer from this PSOO
violation equally, i.e., compositional PMOO analysis cannot outperform
the ULP in this basic scenario~\cite{BS16-1}.

However, in compFFA, there are additional causes for flow segregation.
These can eventually result in PSOO violations and thus reduce accuracy
of algebraically derived bounds. In pursuit of the PSOO principle,
i.e., the least pessimism attainable while preserving correctness
of worst-case bounds, we attempt to mitigate these. Aggregate bounding
of cross-flows was shown to be such a mitigation strategy~\cite{BS15-2}.
Yet, the PMOO's $\beta_{\mathcal{T}}^{\text{l.o.}}$-derivation enforces
a segregation of cross-flows that interfere with the foi on different
subpaths. This is the case for $x\!f_{1}$ and $x\!f_{2}$ in Figure~\ref{fig:BetaLo_PMOO}.
We obtain 
\[
\beta_{\left\langle s_{1},s_{2}\right\rangle }^{\text{l.o.foi}}\,=\,\beta_{\left\langle s_{1},s_{2}\right\rangle }\ominus\left(\alpha_{s_{1}}^{x\!f_{1}},\alpha_{s_{1}}^{x\!f_{2}}\right)\,=\,\beta_{\left\langle s_{1},s_{2}\right\rangle }\ominus\left(\alpha^{x\!f_{1}}\oslash\beta_{s_{0}}^{\text{l.o.}x\!f_{1}},\alpha^{x\!f_{2}}\oslash\beta_{s_{0}}^{\text{l.o.}x\!f_{2}}\right).
\]
The foi sees overly pessimistic cross-traffic arrival bounds compared
to a hypothetical version that aggregates cross-traffic during arrival
bounding. AlgDNC does not provide a tandem $\beta_{\mathcal{T}}^{\text{l.o.}}$
achieving $\beta_{\left\langle s_{1},s_{2}\right\rangle }\ominus\left(\alpha_{s_{1}}^{\left[x\!f_{1},x\!f_{2}\right]}\right)=\beta_{\left\langle s_{1},s_{2}\right\rangle }\ominus\left(\left(\alpha^{x\!f_{1}}+\alpha^{x\!f_{2}}\right)\oslash\beta_{s_{0}}\right)$
in Fig.~\ref{fig:Sample-server-topology}. Other DNC tandem analyses
implementing the PMOO principle for arbitrary multiplexing servers,
namely (min,+) multi-dimensional convolution~\cite{Bouillard_OptRoutingJournal,Bouillard:2008:CMC:1536956.1537043}
and OBA~\cite{Jens_ArbMuxLurch,KGS10-1}, also require segregate
bounding of cross-traffic and thus cause the same PSOO violation. 

The SFA looses awareness of the foi's path in its single-server analysis.
Therefore, it may even segregate cross-flows on this tandem. On the
other hand, it can aggregate all cross-flows sharing a single hop
with the foi. Both situations are illustrated in Figure~\ref{fig:BetaLo_SFA}:
Deriving $\beta_{s_{1}}^{\text{l.o.foi}}$ can use the cross-traffic
aggregate $\left[x\!f_{1},x\!f_{2}\right]$. $\beta_{s_{2}}^{\text{l.o.foi}}$'s
computation enforces the a PSOO violation, indicated by the presence
of $\beta_{s_{0}}^{\text{l.o.}x\!f_{1}}$ and $\beta_{s_{0}}^{\text{l.o.}x\!f_{2}}$,
like in the PMOO equation:
\[
\beta_{\left\langle s_{1},s_{2}\right\rangle }^{\text{l.o.foi}}\,=\,\beta_{s_{1}}^{\text{l.o.foi}}\otimes\beta_{s_{2}}^{\text{l.o.foi}}\,=\,\left(\beta_{s_{1}}\ominus\left(\alpha^{\left[x\!f_{1},x\!f_{2}\right]}\oslash\beta_{s_{0}}\right)\right)\otimes\left(\beta_{s_{2}}\ominus\left(\alpha^{x\!f_{2}}\oslash\left(\beta_{s_{0}}^{\text{l.o.}x\!f_{2}}\otimes\left(\beta_{s_{1}}\ominus\left(\alpha^{x\!f_{1}}\oslash\beta_{s_{0}}^{\text{l.o.}x\!f_{1}}\right)\right)\right)\right)\right)
\]

\subsubsection{Restricting Design of compFFA}

\label{subsec:ProblemDesign-compFFA}The above compFFA's strict separation
of its two parts is also reflected in the algDNC equations. Both,
SFA of Figure~\ref{fig:BetaLo_SFA} and PMOO of Figure~\ref{fig:BetaLo_PMOO},
assume that tandems cannot reach over the foi's path on the left.
A global view on the network reveals that another decomposition is
permissible without violating worst-case assumptions. I.e., we can
match tandems onto the network in a different way:
\begin{center}
\begin{figure}[H]
\begin{centering}
\vspace{-3.5mm}
\includegraphics[width=0.5\columnwidth]{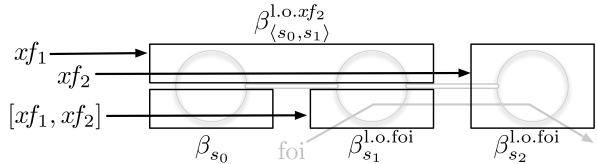}\vspace{-3mm}
\par\end{centering}
\caption{\label{fig:Fig4Alt3}Third alternative to match tandems onto Figure~\ref{fig:Sample-server-topology}.}
\vspace{-8.5mm}
\end{figure}
\par\end{center}

The according (min,+)-equation shows a reduction of enforced segregation
in an SFA-like derivation:
\[
\beta_{\left\langle s_{1},s_{2}\right\rangle }^{\text{l.o.foi}}\,=\,\beta_{s_{1}}^{\text{l.o.foi}}\otimes\beta_{s_{2}}^{\text{l.o.foi}}\,=\,\left(\beta_{s_{1}}\ominus\left(\alpha^{\left[x\!f_{1},x\!f_{2}\right]}\oslash\beta_{s_{0}}\right)\right)\otimes\Bigg(\beta_{s_{2}}\ominus\Bigg(\alpha^{x\!f_{1}}\oslash\underbrace{\left(\beta_{\left\langle s_{0},s_{1}\right\rangle }\ominus\left(\alpha^{x\!f_{2}}\right)\right)}_{\beta_{\left\langle s_{0},s_{1}\right\rangle }^{\text{l.o.}x\!f_{2}}}\Bigg)\Bigg)
\]
where $\beta_{\left\langle s_{0},s_{1}\right\rangle }^{\text{l.o.}x\!f_{2}}$
stretches over the boundaries of both compFFA parts. Thus, with global
knowledge about the network, we can identify more degrees of freedom
when decomposing the network into a sequence of tandem analyses. In
the above example, this can result in arbitrarily better delay bounds
(see Appendix~\ref{sec:Appendix-Quality-algDNC}), similar to the
results presented in~\cite{BS16,BS16-1}. Yet, in Section~\ref{sec:Evaluation}
we evaluate our finding in a broader scope than a selected extreme
case. Moreover, we add a cost evaluation as our finding demands additional
analysis effort to be exploited.
\begin{figure}[b]
\begin{centering}
\subfloat[\label{fig:Sample-server-topology-1}Minimal network with decomposition
effects.]{\begin{centering}
\includegraphics[width=0.45\textwidth]{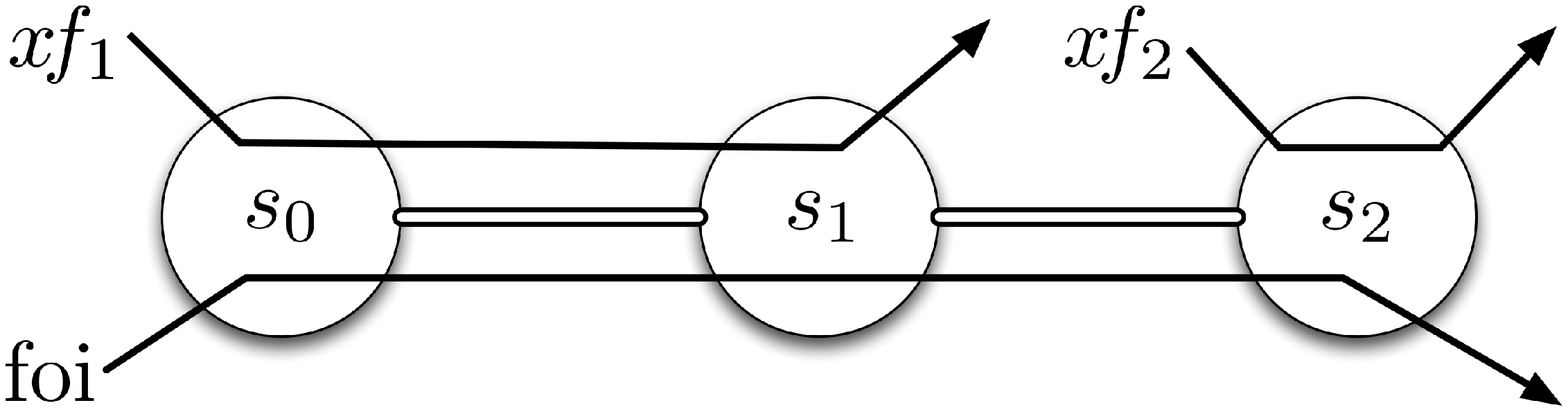}
\par\end{centering}
}
\par\end{centering}
\begin{centering}
\subfloat[\label{fig:BetaLo_PMOO-1}Construction of PMOO's (min,+)-equation.]{\begin{centering}
\includegraphics[width=0.45\textwidth]{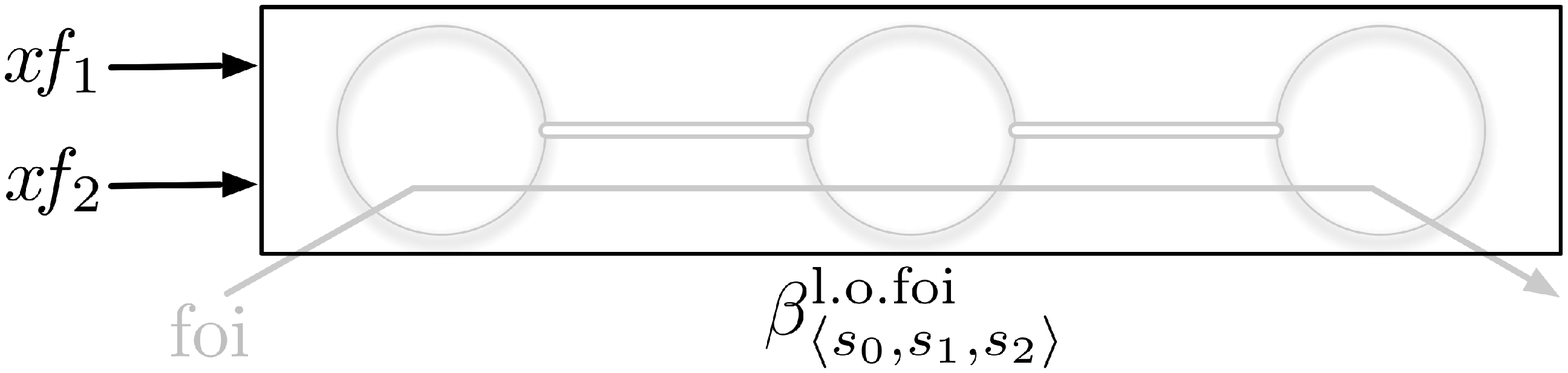}
\par\end{centering}
\centering{}}\hspace*{10mm}\subfloat[\label{fig:BetaLo_SFA-1}Construction of SFA's (min,+)-equation.]{\begin{centering}
\includegraphics[width=0.45\textwidth]{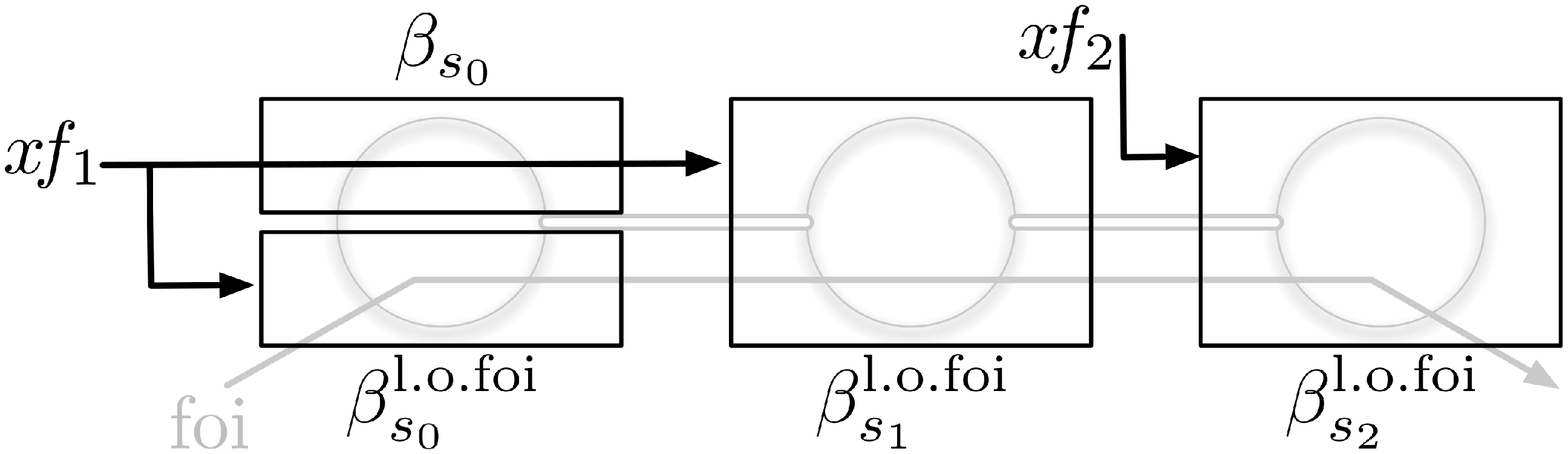}
\par\end{centering}
}
\par\end{centering}
\raggedright{}\caption{Network and derivations of (min,+)-equations illustrating the restricting
design of tandem analyses (Section~\ref{subsec:ProblemDesign-TandemA}).}
\end{figure}

\subsubsection{Restricting Design of Tandem Analyses}

\label{subsec:ProblemDesign-TandemA}Currently, decomposition of the
network is defined by the tandem analysis to be used; it is independent
of the compFFA procedure. We illustrate this on another small sample
network (Figure~\ref{fig:Sample-server-topology-1}). As already
mentioned, PMOO decomposes into longest possible tandems (see Figure~\ref{fig:BetaLo_PMOO-1})
and SFA decomposes into smallest possible tandems (see Figure~\ref{fig:BetaLo_SFA-1}).
Note, that there are no segregation effects in these equations for
tandem analysis. In~\cite{Jens_ArbMuxLurch}, it was shown that the
SFA can outperform the PMOO if the foi's left-over service strongly
differs between servers. The authors designed the first optDNC optimization
formulation to overcome this problem~\textendash{} a tandem analysis
that integrates with the compFFA \textendash{} but its cost renders
an application in feed-forward networks infeasible~\cite{KGS10-1}.
Therefore, we propose to stay with algDNC but decouple the decomposition
of a network into tandems from their analysis. We allow for the following
superior derivation that can exploit the PMOO principle at the start
of the foi's path as well as a fast left-over service curve $\beta_{s_{2}}^{\text{l.o.foi}}$
at the end:
\begin{center}
\includegraphics[width=0.5\columnwidth]{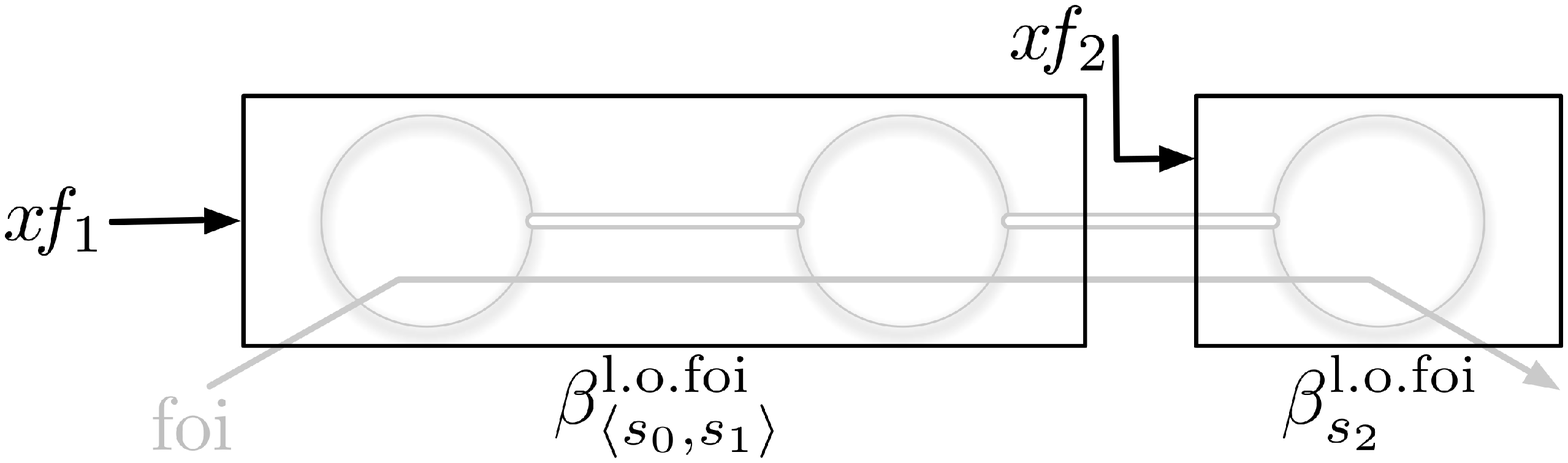}
\par\end{center}

\subsubsection{Decomposition on Insufficient Information}

\label{subsec:DecomposeInsufficientInfo}We presented the previous
two insights to already reveal countermeasures to the problems of
algDNC. However, it is not possible to derive the best tandem decomposition
from the network model in a static fashion. Neither of the alternatives
is strictly superior to the others, i.e., all alternatives need to
be assessed based on their actual cross-traffic arrival bounds and
left-over service curves. Depending on the analyzed tandem's location
in the feed-forward network, the required preceding computations can
already be very costly. Starting with a single decomposition (e.g.,
SFA or PMOO) and eventually recognizing that the analysis can be improved,
requires additional feed-forward analyses for the decomposition alternatives.

\subsubsection{Overly-pessimistic Cross-traffic Burstiness~\cite{BS16}}

\label{subsec:PessimisticBurstsMMB}A recent improvement to algDNC
was proposed in~\cite{BS16}. This work identifies a problem that
occurs when the burstiness of a flow is bounded after crossing a server,
i.e., when applying Theorem~\ref{thm:Performance-Bounds}, Output,
to compute the output bound's relevant part $\alpha'(0)=\left(\alpha\oslash\beta^{\text{l.o.}}\right)\!(0)$.
The computation relies on the deconvolution of the analyzed flow's
arrival curve $\alpha$ with the crossed server's left-over service
curve for it, $\beta^{\text{l.o.}}$. This deconvolution-based computation
is ignorant about the server's queue. It was shown that this can lead
to overly-pessimistic bounds on the worst-case burstiness. During
compFFA's step i), this inaccuracy propagates through the cross-traffic
arrival bounding and eventually leads to overly-pessimistic flow delay
bounds. In~\cite{BS16}, it is propose to derive the worst-case queue
length at a server and use this information to cap the output burstiness~$\alpha'(0)$.
This finding suggests that neither solution to compute bounds on output
burstiness is tight. However, it was only theoretically evaluated
in an artificial network designed to provoke the burst cap mechanism.
An evaluation in realistic networks, benchmarks against optDNC and
an evaluation of the additional effort are lacking.

\subsection{Derivation of all Permissible (min,+)-Equations}

\label{subsec:ExhaustiveDerivation123}Given the new insights about
compositional algDNC analysis we derived in this section, we propose
to exhaustively derive all permissible (min,+)-equations. An equation
is permissible if the result constitutes a valid worst-case delay
bound for the foi. In theory, we execute the following three steps
that somewhat resemble the optDNC feed-forward analysis:
\begin{enumerate}
\item \emph{Backtracking}: Starting at the foi's sink server, this step
derives the dependencies of flows on each other. The backtracking
progresses along the paths taken by flow aggregates in order to mitigate
PSOO violations as much as possible.
\item \emph{(min,+)-equations}: The previous step's backtracking considers
flow entanglements but still leaves the degrees of freedom presented
in Section~\ref{subsec:Feed-forward-Analyses}. We enumerate all
alternatives to match tandems onto the intermediate network representation
from step~1. Tandems are analyzed with the PMOO and interfaced with
the burst-capped output bound. Thus, we derive all permissible (min,+)-equations
that bound the foi's delay in the given feed-forward network.
\item \emph{Finding the best solution}: All equations are permissible and
derive a valid delay bound. The minimum of all solutions is the most
accurate valid delay bound.
\end{enumerate}
Note, that the SFA and the PMOO analysis are algDNC heuristics that
reduce the computational cost of steps~2 and~3 by reducing the search
space to a single tandem matching / (min,+)-equation~\textendash{}
just like the ULP only derives a single linear program. Therefore,
our exhaustive decomposition scheme is guaranteed to be at least as
good as the SFA and the PMOO analysis.

Whereas optDNC's LP analysis searched for the worst delay of any linear
program, we search for the best delay bound of any algebraic (min,+)-equation.
Our search space, the set of permissible equations bounding the foi's
delay, is constructed subject to the result of the backtracking step
and the composition rules for algDNC operators. Therefore, our search
follows the objective to find the permissible algDNC (min,+)-equation
that
\begin{itemize}
\item maximizes aggregation of cross-flows~\cite{BS15-2},
\item minimizes the impact of enforced segregation,
\item maximizes benefits of long tandems,
\item maximizes gain from fast local left-over service curves, and
\item caps overly pessimistic cross-traffic burstiness~\cite{BS16}.\vspace{2mm}
\end{itemize}
As our exhaustive procedure's design somewhat resembles the optDNC's
LP, an obvious concern is its computational cost. As a matter of fact,
the amount of decompositions and the involved operations increase
with the network size, reaching $10^{9}$ fast. Figure~\ref{fig:Combinatorial-compFFA}
illustrates the combinatorial explosion on a small network with just
$32$ devices. Routing $100$ flows over this network, the worst-case
amount of permissible (min,+)-equations of a single flow reaches close
to $\text{10}^{20}$, and for $300$ flow, more than $\text{10}^{100}$
equations exist for some flows. Next, we design an efficient algorithm
that allows to attain the same delay bound quality as this total enumeration
approach, yet, it imposes only a tiny fraction of the cost.
\begin{figure}[H]
\begin{centering}
\includegraphics[width=0.55\textwidth]{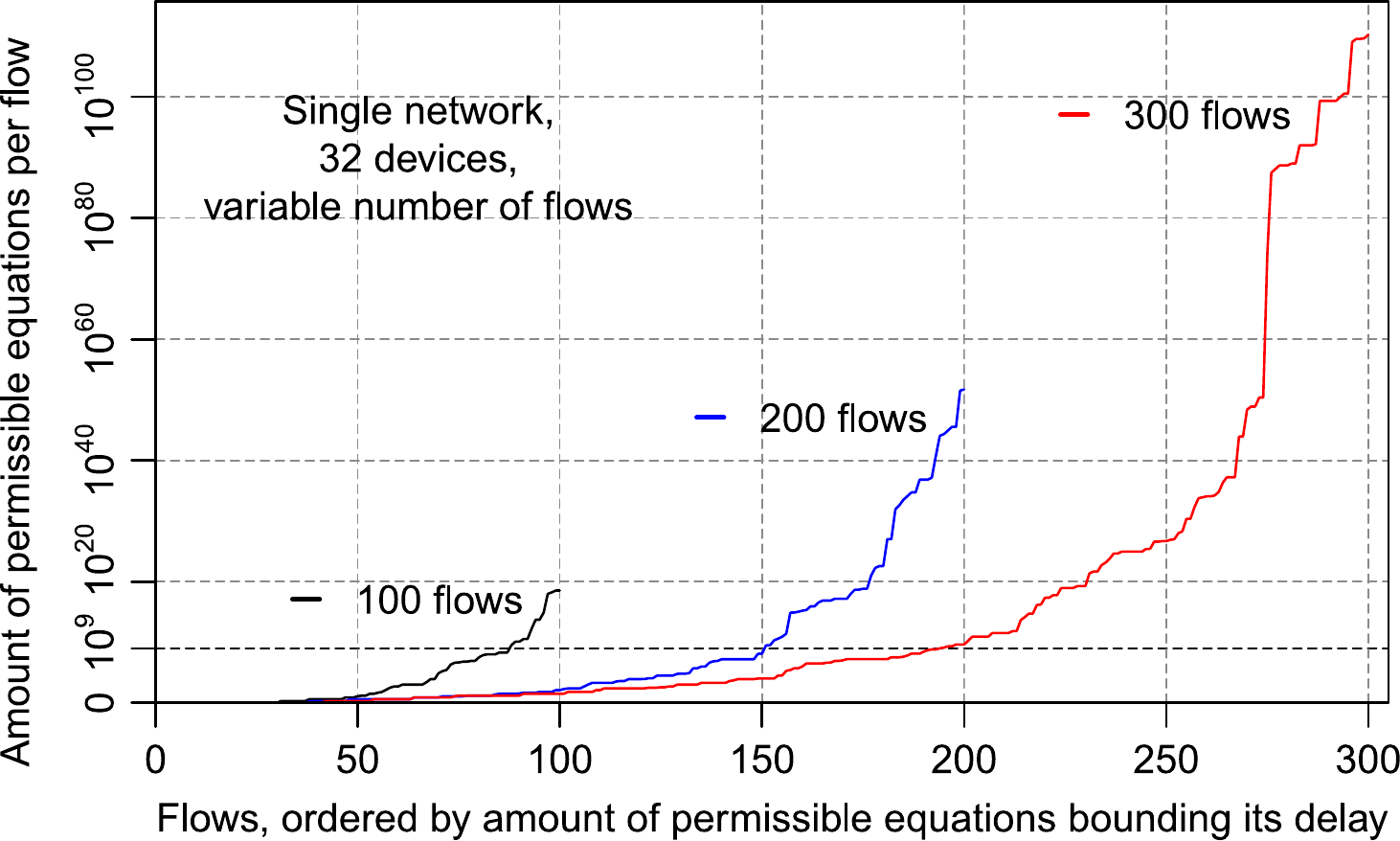}
\par\end{centering}
\caption{\label{fig:Combinatorial-compFFA}Combinatorial explosion in exhaustive
enumeration of all permissible (min,+)-equations with the compFFA
procedure.}
\end{figure}

\section{Exhaustiveness \& Efficiency \textendash{} A Novel Algorithm Design
for DNC}

\label{sec:algDNCwithOpt}In this section, we provide an algorithm
that efficiently computes the best delay bound attainable with the
exhaustive set of permissible (min,+)-equations. In Section~\ref{sec:Evaluation},
we will benchmark it against optDNC.

\subsection{Integrated Design}

We depart from separately deriving global knowledge and the permissible
(min,+)-equations. Instead, we unify steps~1 and~2 of Section~\ref{subsec:ExhaustiveDerivation123}'s
compFFA into a single one Step~1, the backtracking, was designed
to derive information of the entanglement of flow aggregates and step~2
used this information derive a permissible (min,+)-equation based
on the composition rules for (min,+)-operations. In this Section,
we show that it is possible to unify both steps, yet, still consider
all permissible equations and derive the best valid result.

We aim to derive all decompositions of a feed-forward network in a
single step. The exhaustiveness we propose is defined by a multitude
of alternatives that were previously considered in step~2. In order
to allow our algorithm to attain the same set of permissible equations,
we need to extend the new backtracking to consider alternatives as
well. The solution is simple yet powerful. The analysis starts with
a single tandem (the foi's path) and recursively backtracks cross-traffic
flow aggregates over the tandems they jointly traverse until reaching
the tandem currently analyzed. Given any tandem in this procedure,
the integrated analysis decomposes it into all disjoint sequences
of sub-tandems by deciding whether to decompose (``cut'') the tandem
at a link or not. On sub-tandems, the left-over service curve is computed
and all sub-tandem ones are convolved into the tandem one for a specific
decomposition alternative. In order to derive left-over service curves,
the backtracking needs to be started each of sub-tandem, exhaustively
working on decompositions again. The exhaustive derivation all sub-tandem
decompositions mitigates the causes for compFFA's pessimism:
\begin{description}
\item [{Section~\ref{subsec:EnforcedSegregations}}] Compositions that
align with sub-path sharing of cross-flows are tested. They trade
algDNC's need for segregated bounding against composite left-over
derivation.
\item [{Section~\ref{subsec:ProblemDesign-compFFA}}] Small sub-tandems
on the foi's path are combined with large (sub-)tandems for cross-traffic
arrival bounding.
\item [{Section~\ref{subsec:ProblemDesign-TandemA}}] Neglected decompositions
for the foi's are included.
\item [{Section~\ref{subsec:DecomposeInsufficientInfo}}] Sufficiently
many decompositions are tested to find the best one for any tandem
in the analysis.
\item [{Section~\ref{subsec:PessimisticBurstsMMB}}] The output of each
tandem can be easily capped with the backlog bound of its last server.
\end{description}
In fact, as the backtracking simultaneously considers the entanglement
of flows (backtracking of flow aggregates) and the composition rules
of (min,+)-operations, it derives the same set of permissible equations
as presented in Section~\ref{subsec:ExhaustiveDerivation123}.

\subsection{Efficiency Concerns and Improvements}

The exhaustiveness of our proposed procedure naturally raises concerns
about its computational costs. A tandem of length $n$ can be decomposed
into $2^{n-1}$ distinct sub-tandem sequences. All sub-tandem sequence
of two adjacent levels of the backtracking recursion need to be combined
in order to find the best attainable algDNC result. This consideration
reveals a combinatorial explosion resulting in the large amount of
permissible equations presented in Figure~\ref{fig:Combinatorial-compFFA}.
In detail, a large search tree (Figure~\ref{fig:Search-tree-Sample})
is constructed where branching corresponds to the sub-tandem decomposition
that, in turn, requires bounding of cross-traffic arriving on a tandem
to cut. Each path through this search tree derives a single permissible
(min,+)-equation. With our integrated design, we aim for an efficient
solution to fully cover the search tree and achieve the objective
to find the permissible equation for most accurate results. In algDNC,
we can compute intermediate results by solving already derived parts
of permissible (min,+)-equations. Their results can be used to cut
down the search tree and thus improve efficiency.
\begin{figure}
\begin{centering}
\includegraphics[width=0.425\textwidth]{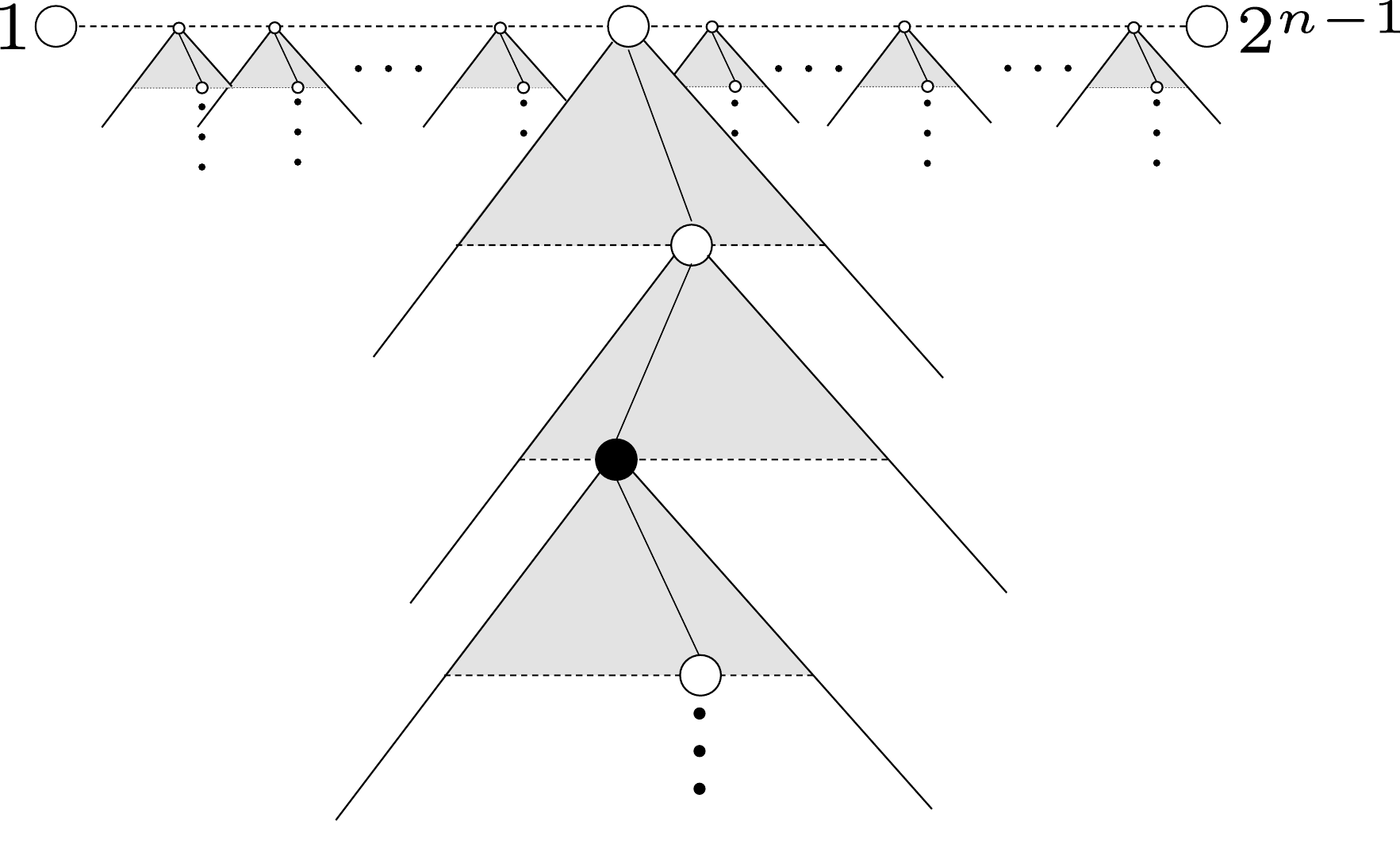}
\par\end{centering}
\caption{\label{fig:Search-tree-Sample}Search tree \cite{SearchTree} for
permissible (min,+)-equations. The black dot marks a computation of
cross-traffic arrivals.}
\end{figure}

\subsubsection{Caching of Arrival Bounds}

Disjoint sub-tandem decompositions only need to differ by a single
link that was cut, so they often share many sub-tandems which all
require the same derivations of traffic arrivals. These arrival bounds
are solely defined by the sub-tandem, the flows to bound, and the
foi~\textendash{} i.e., the path through the search tree that leads
to a certain arrival bound (black dot in Figure~\ref{fig:Search-tree-Sample})
does not matter. Caching and reusing arrival bounds therefore reduces
the computational effort by allowing the analysis to terminate a search
before reaching any of the search tree's leaves.

\subsubsection{Convolution of Alternative Arrival Bounds}

In the search tree, effort spreads over adjacent levels: Each decomposition
results in one left-over service curve that computes one arrival bound,
all of which are combined with those of the next higher level. In
an algDNC analysis, we compute these arrival bounds as internal, intermediate
results and will use them to counteract the combinatorial explosion.
All of these bounds are valid arrival curves and therefore their
convolution is a valid arrival curve as well~\cite{LeBoudec_NCbook}.
This countermeasure reduces the amount of arrival bounds per level
in the recursion to a single one. Hence, it prevents the combinatorial
explosion when combining any two adjacent recursion levels' results.

\subsection{The Exhaustive yet Efficient Decomposition Algorithm}

\label{subsec:Algorithm}Our efficient algorithm depicted in Algorithm~\ref{alg:OnlineAlgorithm}
achieves very accurate compositional algDNC delay bounds in feed-forward
networks. It computes a flow of interest's alternative end-to-end
left-over service curves using the following steps: First, ${\tt get}\mathbb{B}_{\mathcal{T}}^{{\tt l.o.}\mathbb{F}}$
is invoked with the foi's path $\mathcal{P}$, i.e., $\mathcal{T}\coloneqq\mathcal{P}$,
and the foi itself. $\mathcal{P}$ constitutes the first tandem to
decompose into $2^{n-1}$ disjoint sub-tandem decompositions $\{\mathbb{T}_{1},\ldots,\mathbb{T}_{2^{n-1}}\}$,
where $n$ is the number of servers on $\mathcal{P}$. For each sub-tandem
$T$ of such a decomposition, we next backtrack (forcibly segregated,
Section~\ref{subsec:EnforcedSegregations}) cross-traffic to bound
its arrivals. Our efficiency improvements can be found in this part
of the algorithm: retrieval of a cached bound (line 20), convolution
of alternative bounds (line~35) and caching of a bound (line~42).
After the arrival bounding recursion has terminated, its results
are used to derive the foi's left-over service curve for composition
alternative $\mathbb{T}$, $\beta_{\mathbb{T}}^{\text{l.o.\ensuremath{\mathbb{F}}}}$.
The efficiency of this step could be increased by caching a sub-tandem
$\beta_{T}^{\text{l.o.}}$ which we leave for future work. In a final
step (not shown in Algorithm~\ref{alg:OnlineAlgorithm}) the delay
bounds for all decomposition alternatives are computed. The minimum
among them is the best algDNC delay bound attainable with our exhaustive
decomposition design. In Appendix~\ref{sec:Computational-complexity-algDNC}
we provide more details about our algDNC algorithm's complexity and
compare it to the SFA from the literature.

\IncMargin{0.75mm}

\begin{algorithm2e}
\caption{\label{alg:OnlineAlgorithm} Exhaustive Decomposition Algorithm:
\mbox{Efficient} Derivation of all Left-over Service Curves of $\mathcal{T}$.}

\SetKwProg{Fn}{}{}{}
\SetKwFunction{CUT}{getDecompositions}
\SetKwFunction{FsharedPathTo}{backtrackTandem}
\SetKwFunction{BETAS}{get$\mathbb{B}_{\mathcal{T}}^{\text{l.o.}\mathbb{F}}$\!}
\SetKwFunction{GROUP}{xtxSegregation}
\SetKwFunction{ADD}{put}
\SetKwFunction{GETCACHE}{getCacheEntry}
\SetKwFunction{ADDCACHE}{addCacheEntry}
\SetKwFunction{SRCACs}{getSrcFlow$\alpha$s}
\SetKwFunction{TFA}{getBurstCap}
\SetKwFunction{BurstCap}{cap}
\SetKwFunction{FABs}{AB}
\SetKwFunction{Finlink}{splitAsPerInlink}
\SetKwFunction{FgetSource}{getSrc}
\SetKwFunction{FgetSink}{getSink}
\AlgoDisplayBlockMarkers\SetAlgoBlockMarkers{}{}
\SetAlgoNoEnd

\Fn(){\BETAS{\emph{Tandem} $\mathcal{T}$, \emph{Flow aggregate of interest} $\mathbb{F}$}}{
 /* Disjoint subtandem decompositions of $\mathcal{T}$ */ \\
 $\{\mathbb{T}_{1},\ldots,\mathbb{T}_{2^{n-1}}\}$ = \CUT{$\mathcal{T}$}\!;\\
 \ForEach{Decomposition $\mathbb{T}\in\{\mathbb{T}_{1},\ldots,\mathbb{T}_{2^{n-1}}\}$}{
  \ForEach{Subtandem $T\in\mathbb{T}$}{
   /* Enforced segregations (Section~\ref{subsec:EnforcedSegregations}) */\\
   $\left(\mathbb{F}_{i},\ldots,\mathbb{F}_{i+j} \right)$ = \GROUP\!\!$(T,\mathbb{F})$;\\
   \BlankLine
   /* Arrival bounding and left-over service*/\\
   $(\alpha^{\mathbb{F}_{i}}_{T},\ldots,\alpha^{\mathbb{F}_{i+j}}_{T})$ = \FABs\!\!$(T,\left(\mathbb{F}_{i},\ldots,\mathbb{F}_{i+j} \right))$;\\
   $\beta_{T}^{\text{l.o.}\mathbb{F}}$ = $\beta_{T}\ominus(\alpha^{\mathbb{F}_{i}}_{T},\ldots,\alpha^{\mathbb{F}_{i+j}}_{T})$;\\
   \BlankLine
   /* End-to-end service of decomposition $\mathbb{T}$ */\\
   $\beta_{\mathbb{T}}^{\text{l.o.}\mathbb{F}}$ $\otimes$= $\beta_{T}^{\text{l.o.}\mathbb{F}}$;\\
  }
  \textbf{end}\\
  $\mathbb{B}_{\mathcal{T}}^{\text{l.o.}\mathbb{F}}$.\ADD$\!\!(\beta_{\mathbb{T}}^{\text{l.o.}\mathbb{F}})$;\\
 }
 \textbf{end}\\
}
\Return $\mathbb{B}_{\mathcal{T}}^{\text{l.o.}\mathbb{F}}$

\BlankLine
\Fn(){\FABs\!\!(\emph{Tandem} $T$, \emph{Flow\,aggregates\,to\,bound\,$\left(\mathbb{F}_{i},\ldots,\mathbb{F}_{i+j} \right))$}}{
 \ForEach{Flow aggregate $\mathbb{F} \in \left(\mathbb{F}_{i},\ldots,\mathbb{F}_{i+j} \right)$}{
  /* Efficiency: Check for cached $\alpha^{\mathbb{F}}_{T}$ */\\
  $\textbf{try}\{\alpha^{\mathbb{F}}_{T}=\GETCACHE\!\left(T,\,\mathbb{F}\right)$;\\
\hspace{6.5mm}$\mathbb{A}_{T}^{\mathbb{F}}.\ADD\left(\alpha^{\mathbb{F}}_{T}\right)$;\\
\hspace{6.5mm}\textbf{continue}; \}\\
  /* Recursively backtrack flows in $\mathbb{F}$: */\\
  /* Get $\mathbb{F}$'s tandem before reaching $T$ */\\
  $\mathcal{T}_\text{shared}$ = \FsharedPathTo\!\!($\mathbb{F},T$);\\
  /*\,Compute and store left-over service curves\,*/\\
  $\mathbb{B}_{\mathcal{T}_\text{shared}}^{\text{l.o.}\mathbb{F}}$ = \BETAS\!($\mathcal{T}_\text{shared},\mathbb{F}$);\\
  /*\,Get arrivals of flows in\,$\mathbb{F}$ at $\mathcal{T}_\text{shared}$'s\,source\,*/\\
  $src=\mathcal{T}_\text{shared}.\FgetSource();$  /* $src$ has $l$ inlinks */\\
  $\left(\mathbb{F}_{1},\ldots,\mathbb{F}_{l}\right)= \Finlink(src,\mathbb{F})$;\\
  $\alpha_{\mathcal{T}_\text{shared}}^{\mathbb{F}}$ = $\sum$ \FABs\!\!$(\left\langle src\right\rangle\!,\left(\mathbb{F}_{1},\ldots,\mathbb{F}_{l} \right))$;\\
  \BlankLine
  /* Arrivals at $T$ is their output from $\mathcal{T}_\text{shared}$ */\\
\ForEach{\emph{$\beta_{\mathcal{T}_\text{shared}}^{\text{l.o.}\mathbb{F}}\in\mathbb{B}_{\mathcal{T}_\text{shared}}^{\text{l.o.}\mathbb{F}}$}}{
  /* Efficiency: Convolve alternative $\alpha^{\mathbb{F}}_{T}$s */\\
  $\alpha^{\mathbb{F}}_{T}$ $\otimes$= $(\alpha_{\mathcal{T}_\text{shared}}^{\mathbb{F}}\oslash\beta_{\mathcal{T}_\text{shared}}^{\text{l.o.}\mathbb{F}})$;\\
 }
 \textbf{end}\\
 \BlankLine
 /* Burst cap \cite{BS16} */\\
 $B^{\text{max}}$ = \TFA\!\!($\mathcal{T}_\text{shared}$.\FgetSink\!\!());\\
 $\alpha^{\mathbb{F}}_{T}$.\BurstCap\!\!($B^{\text{max}}$);\\
 $\mathbb{A}_{T}^{\mathbb{F}}$.\ADD\!\!($\alpha^{\mathbb{F}}_{T}$);
 \BlankLine
 /* Efficiency: Cache $\alpha^{\mathbb{F}}_{T}$ */\\
 $\ADDCACHE\!\left(\alpha^{\mathbb{F}}_{T}\right)$;\\
 }
 \textbf{end}\\
}
\Return $\mathbb{A}_{T}^{\mathbb{F}}=(\alpha^{\mathbb{F}_{i}}_{T},\ldots,\alpha^{\mathbb{F}_{i+j}}_{T})$\\

\end{algorithm2e}

\section{Evaluation}

\label{sec:Evaluation}\label{subsec:Numerical-Experiments}In this
numerical evaluation, we benchmark our new algDNC analysis algorithm
against the optDNC's ULP as well as the SFA. It is known that the
ULP will derive the most accurate delay bounds among the alternatives
and that SFA will derive the least accurate ones. We show that our
novel algDNC analysis considerably outperforms the SFA. It even derives
delay bounds close to the ULP, yet, in a fraction of the ULP's computation
time.\begin{figure}
\begin{centering}
\hspace*{-1mm}\subfloat[\label{fig:Eval-AFDX-DelayBounds}AFDX Case Study: Delay Bounds.]{\begin{centering}
\includegraphics[width=0.525\textwidth]{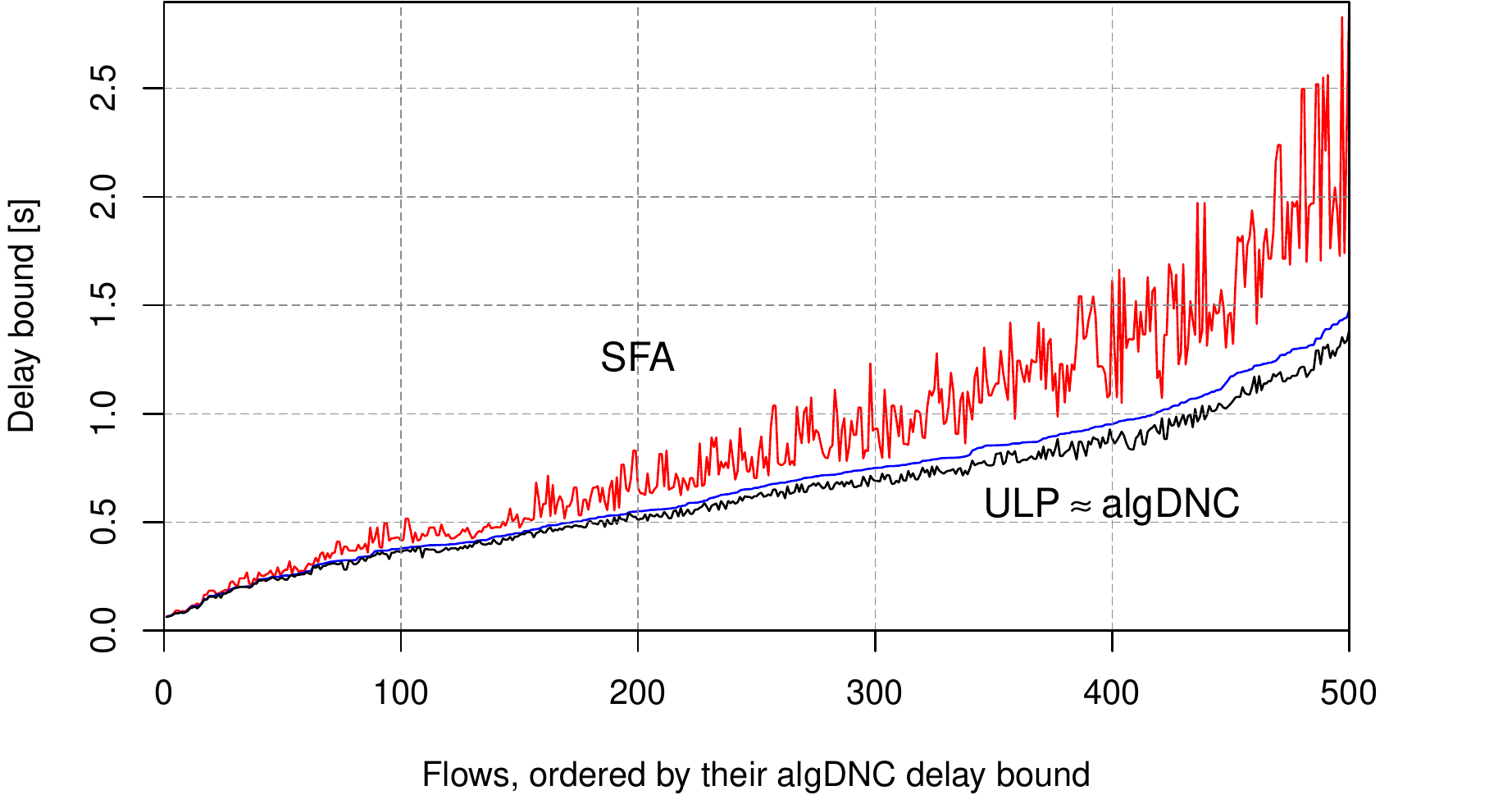}
\par\end{centering}
}\hspace*{-7.5mm}\subfloat[\label{fig:Eval-AFDX-RunTimes}AFDX Case Study: Computational Effort.]{\begin{centering}
\includegraphics[width=0.525\textwidth]{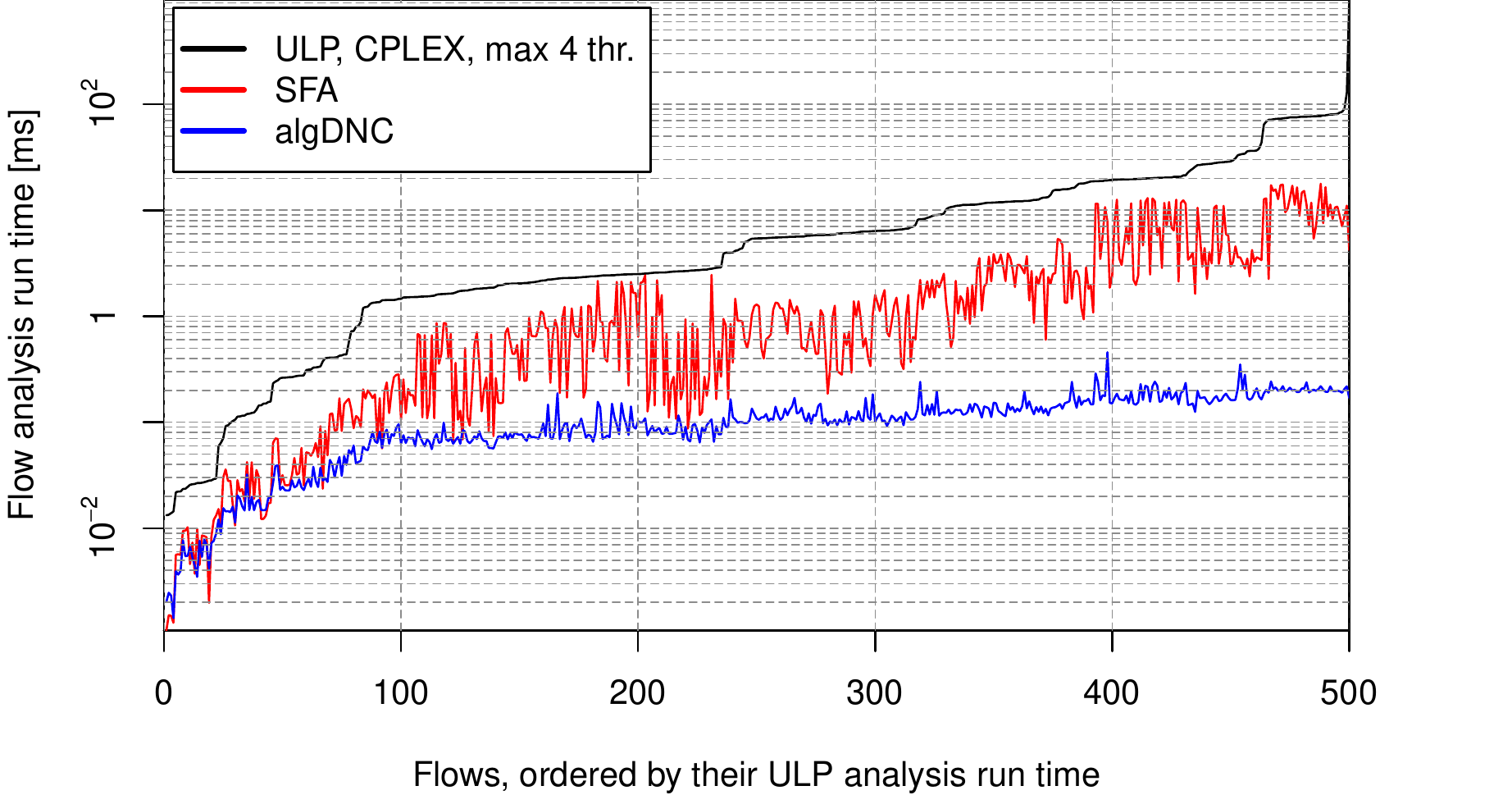}
\par\end{centering}
}
\par\end{centering}
\caption{\label{fig:Eval-AFDX}Delay bounds and computational effort in the
AFDX topology with randomly routed unit size flows.}
\end{figure}

\subsection{AFDX Topology Case Study}

\label{subsec:AFDX-Case-Study}First, we investigate delay bound accuracy
in a network topology as found in the industrial avionics context.
The topology we exemplarily analyze is dimensioned similarly to the
backbone network in the Airbus~A380. It has a dense core of $16$~switches
that connect a total of $125$~end-systems located in the network
periphery. Each server has a service curve resembling a $100$Mbps
Ethernet link. We created a representative AFDX topology according
to the algorithm presented in~\cite{Pegase_ERTSS}. This topology
generation scheme has some random factors in it, i.e., from an industrial
point of view, the network we analyze may correspond to a single alternative
in the pre-deployment design space exploration.

According to the current AFDX specification, flows are routed within
so-called virtual links (VLs). Each VL connects a single source end-system
to multiple sink end-systems (in the device graph) with fixed resource
reservation on the path between these systems. In the view of the
network calculus, VLs correspond to multicast flows that reserve large
resource shares. An examination of the problems due to VLs\textquoteright{}
coarse granularity can be found in~\cite{AFDX_LinkUtil}. The analysis
of multicast flows with algDNC was presented in~\cite{BG16}, yet,
a multicast ULP has not been provided. For an expressive comparison,
we thus focus on the AFDX network topology itself and $500$ randomly
routed unicast flows to the network. They are shaped to unit size
token buckets wit rate $1$Mbps and bucket size $1$Mb.

Figure~\ref{fig:Eval-AFDX-DelayBounds} depicts the \textbf{$500$}
individual flow delay bounds. The SFA's delay bounds show a gap to
algDNC and ULP bounds that tends to grow on average. Additionally,
the SFA delay bounds keep oscillating compared to the ULP and algDNC
bounds, such that this analysis is not suitable to confidently rank
AFDX design alternatives regarding their performance.

Figure~\ref{fig:Eval-AFDX-RunTimes} shows the time to analyze each
flow in the network. This per-flow effort can differ more than three
orders of magnitude within the SFA and the ULP analysis. This is surprising
as the AFDX has a small network diameter and therefore the recursive
backtracking is not very deep. For the very same reason, algDNC's
effort stays within a much smaller range. In absolute terms, the algDNC
outperforms ULP as well as SFA run times by multiple orders of magnitude
\textendash{} a decisive advantage in design space explorations.
\begin{figure*}
\begin{centering}
\subfloat[\label{fig:Delay-bound-eval-DelayBounds-20}Sample network I (20 devices).]{\begin{centering}
\includegraphics[width=0.325\textwidth]{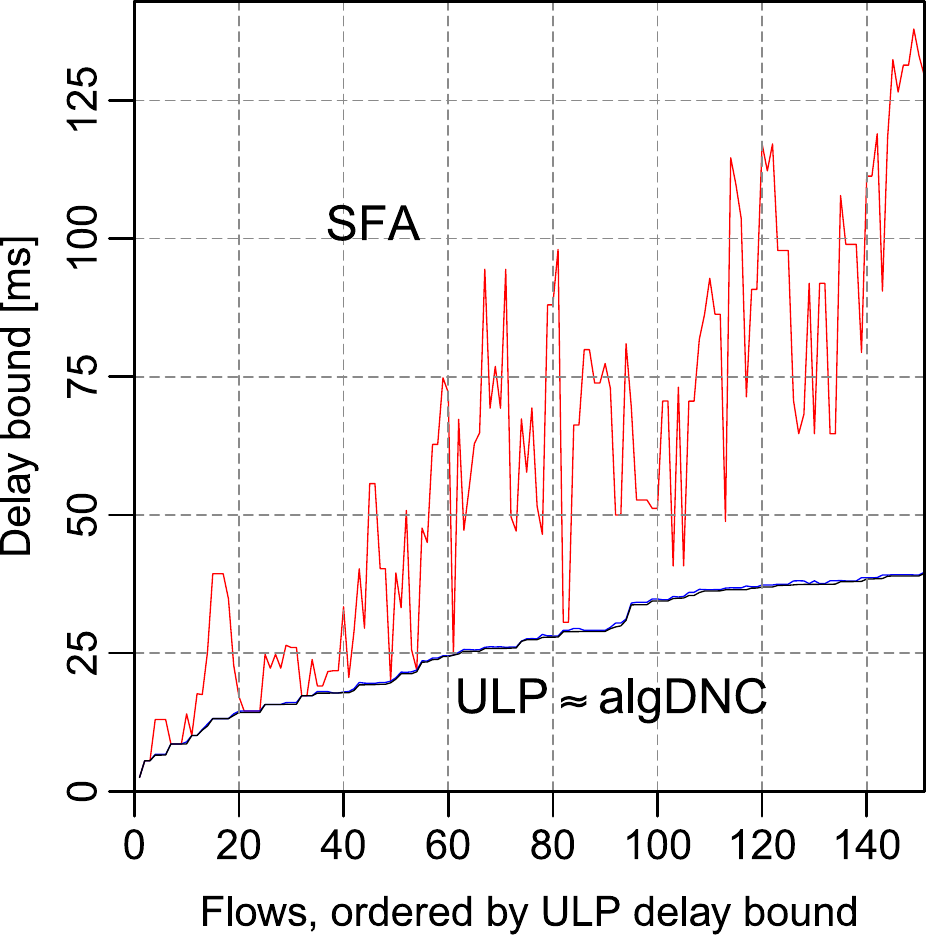}
\par\end{centering}
}\hspace*{10mm}\subfloat[\label{fig:Delay-bound-eval-DelayBounds-180}Sample network II (180
devices).]{\begin{centering}
\includegraphics[width=0.325\textwidth]{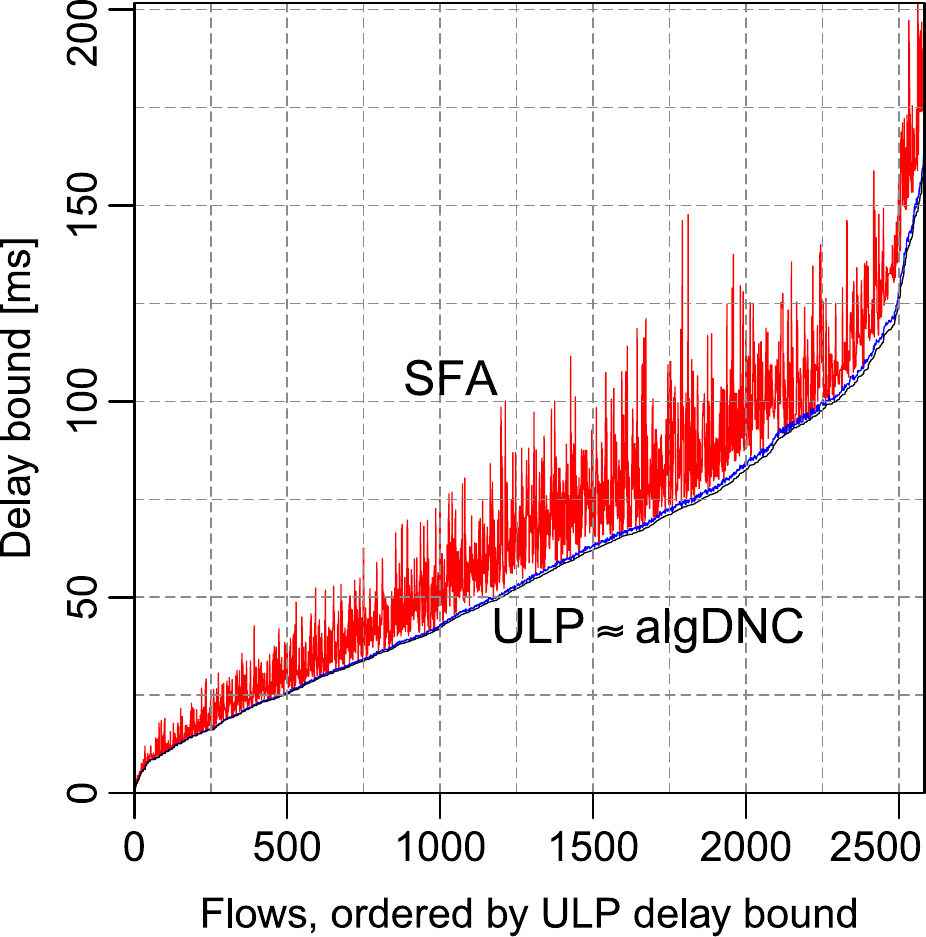}
\par\end{centering}
\centering{}}
\par\end{centering}
\begin{centering}
\subfloat[\label{fig:Delay-bound-eval-DeviationULP}Comparison to ULP (networks
from 20 to 180 devices).]{\begin{centering}
\includegraphics[width=0.75\textwidth]{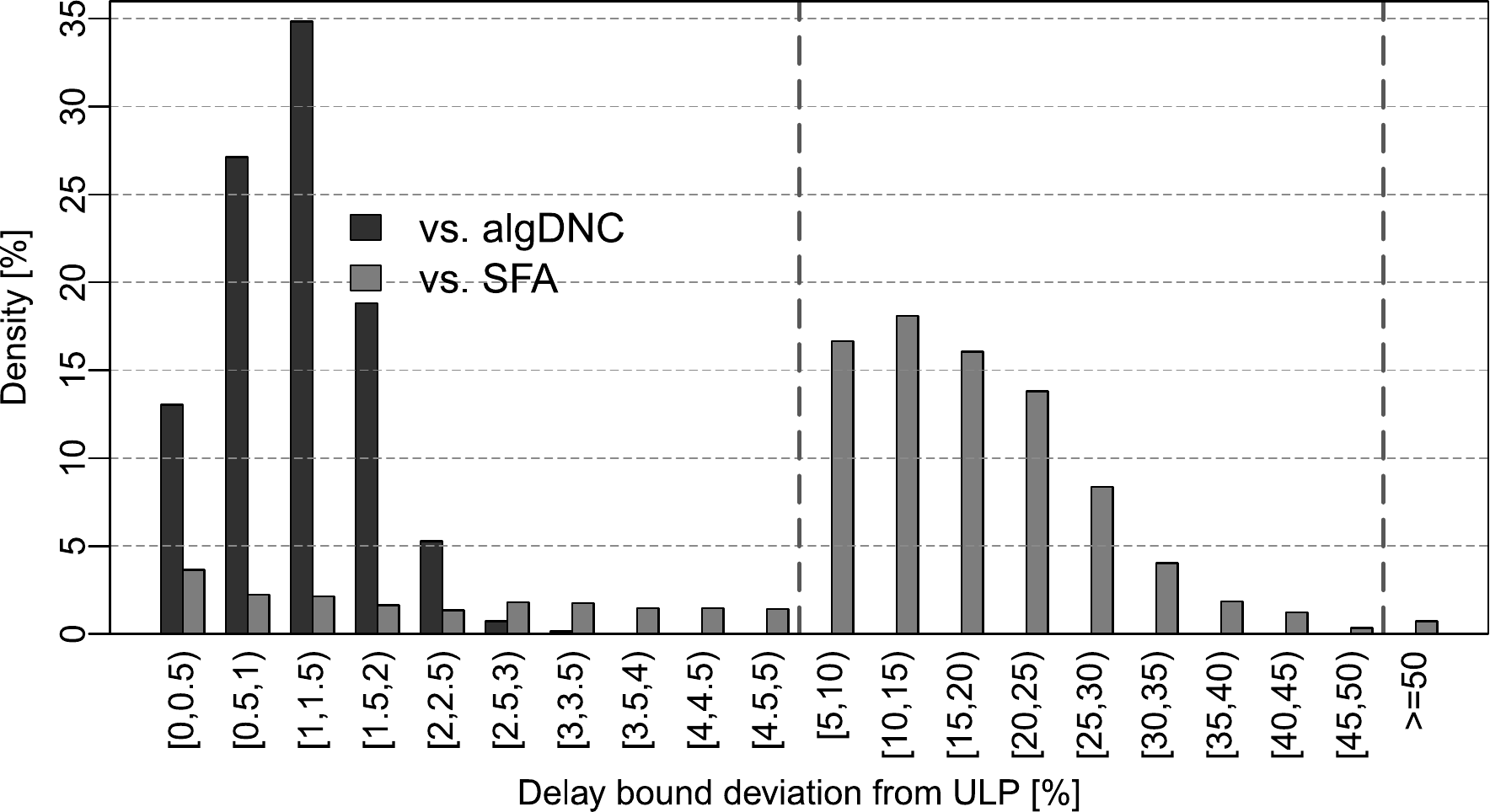}
\par\end{centering}
}
\par\end{centering}
\centering{}\caption{\label{fig:Delay-bound-evaluation}Delay Bound Accuracy. Subfigures~\ref{fig:Delay-bound-eval-DelayBounds-20}
and~\ref{fig:Delay-bound-eval-DelayBounds-180} show the delay bounds
for two sample networks of different sizes. They depict each flow's
end-to-end delay bound computed with SFA, ULP and our new exhaustive
algDNC. While the old SFA cannot compete with the ULP, the results
of our new algDNC analysis are extremely close to the ULP delay bounds.\vspace{1.5mm}
\protect \\
Subfigure~\ref{fig:Delay-bound-eval-DeviationULP} provides quality
statistics over all the networks from $20$ to $180$ devices ($12376$
flows, Table~\ref{tab:NetServerFlow}, left, in Appendix~\ref{subsec:SetNetworks})
by depicting the deviation of delay bounds relative to the ULP's results.
The accuracy of our new algDNC analysis is stable across all network
sizes and, except for a single outlier at $7.57$\%, all new algDNC
bounds deviate from ULP by at most $4.2$\%.}
\end{figure*}

\subsection{Scalability of DNC Analyses}

Next, we turn to the scalability. To do so, we created a set of larger
test networks (see Appendix~\ref{subsec:SetNetworks}) to evaluate
quality and cost.

\subsubsection{Quality: Accuracy of Delay Bounds}

 Figures~\ref{fig:Delay-bound-eval-DelayBounds-20} and~\ref{fig:Delay-bound-eval-DelayBounds-180}
show each flow's end-to-end delay bound in two sample networks of
very different size and complexity: the smallest in our set of networks
($20$ devices, $38$ servers, $158$ flows) and the largest one still
feasible to analyze with ULP ($180$ devices, $646$ servers, $2584$
flows). For the small network, we already showed that the SFA results
oscillate wildly with a large amplitude above the ULP (Section~\ref{subsec:optDNC}).
Figure~\ref{fig:Delay-bound-eval-DelayBounds-20} extends this evaluation
with the algDNC bounds and reveals that they increase in lockstep
with the ULP while staying in close range. The same holds true in
the larger sample network of Figure~\ref{fig:Delay-bound-eval-DelayBounds-180}.
Both observations are also confirmed by the overall results of our
experimental investigation. Figure~\ref{fig:Delay-bound-eval-DeviationULP}
depicts the deviation of the SFA and the algDNC delay bounds from
the ULP as observed in all network sizes from $20$ to $180$ with
a total of $12376$ observation points (Table~\ref{tab:NetServerFlow},
flows on the left, in Appendix~\ref{subsec:SetNetworks}). The numerical
evaluation confirms that optDNC's ULP barely outperforms the algDNC:
\begin{itemize}
\item Our new exhaustive decomposition analysis deviates from the ULP by
only $1.142$\% on average in our evaluations.
\item The $99$th percentile is as low as $2.48$\% deviation.
\item Results do not deviate by more than $4.2$\%, except for a single
outlier at $7.57$\%. Yet, this outlier is still more accurate than
$73.2$\% of the SFA-derived delay bounds.
\item These accuracy characteristics are stable across different network
sizes.
\end{itemize}
These results confirm that the problems of algDNC found in Section~\ref{subsec:Problems-of-algDNC}
were indeed crucial for its previous lack in accuracy. Moreover, our
approach to find the (min,+)-algebraic equation with the minimal combined
impact of all algDNC problems turned out to allow for delay bounds
whose quality is competitive with optDNC's feasible heuristic ULP.

\subsubsection{Cost: Network Analysis Times}

Last, we evaluate the cost of attaining delay bounds in feed-forward
networks. We show that our efficiency improvements have a crucial
impact on the network analysis times of the algDNC analysis. Figure~\ref{fig:AnalysisEffort}
shows that both outperform the ULP considerably and scale better with
increasing network size. Our novel algDNC analysis performs best,
a discernible increase in effort can only be observed in network sizes
that are impractical to analyze with any of the other analyses. Although
the experimental results are subject to fluctuations due to the high
degree of randomness in the creation of our set of sample networks
(see Appendix~\ref{subsec:SetNetworks}), we can draw clear conclusions
about fundamental trends in DNC network analysis:
\begin{itemize}
\item The ULP becomes computationally infeasible fast; the $180$ devices
network requires $\sim\!13$~days to be analyzed.
\item The SFA scales better than the ULP. However, absolute effort increases
to levels unsuitable for a design space exploration. Increasing the
network size by only $20$ devices can cause a huge difference of
analysis effort. In our sample, analyzing the $260$ devices network
is relatively fast while the $280$ devices network is already impractical
to analyze. The SFA's effort seems barely predictable.
\item Our novel algDNC scales better than the other analyses. It is also
more resilient to the randomness of our network creation. Provoking
a considerable increase of network analysis time required to vastly
increase the network size to $1000$ devices ($3626$ servers, $14504$
flows).
\end{itemize}
Figure~\ref{fig:AnalysisEffort} also depicts results for different
levels of deterministically parallelized optimization with CPLEX\footnote{Note, that among the tools employed in this evaluation, only CPLEX
offers parallelization.}. A maximum amount of $4$~threads yields larger network analysis
times than single-threaded execution due to the overhead of thread
synchronization. However, for networks $>\!120$~devices, an optimization
with up to $8$~parallel threads becomes faster. In Section~\ref{subsec:optDNC},
Figure~\ref{fig:ULP-Share-DiscoDNC-CPLEX}, we saw that CPLEX optimization
consumes the vast majority of analysis time. In an attempt to improve
optDNC analyses run time, we also investigated the potential benefit
of further parallelization of this optimization step with CPLEX. Yet,
run times were not affected significantly as the results in Appendix~\ref{sec:EvalCostParallelCPLEX}
show. Deriving accurate delay bounds in large feed-forward networks
is solely possible with the novel algDNC analysis. OptDNC's trends
w.r.t. delay bounds of individual flows are retained such that algDNC
possesses a similar power in ranking of design alternatives.\begin{figure*}[!t]
\begin{centering}
\includegraphics[width=0.8\textwidth]{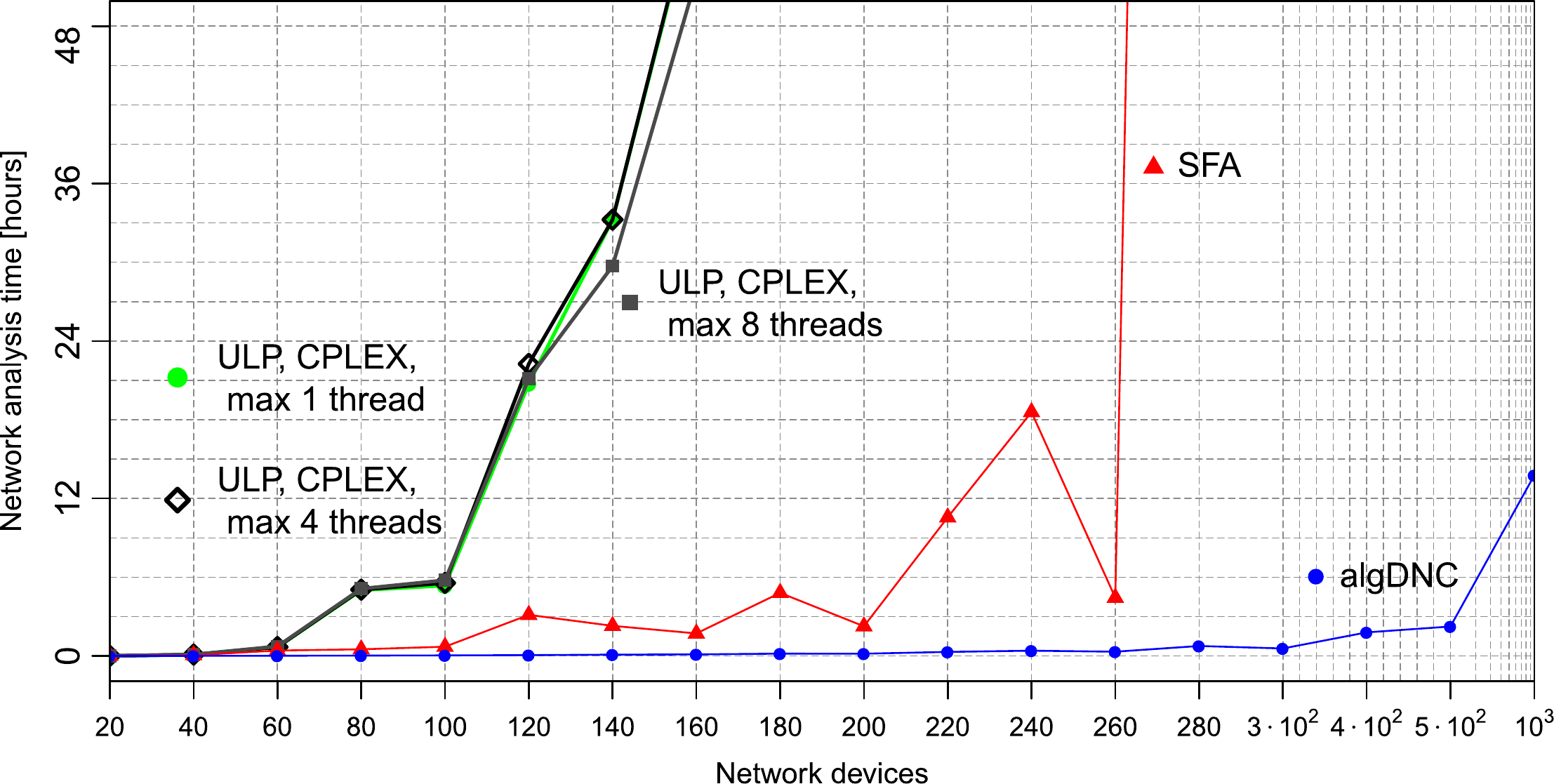}
\par\end{centering}
\vspace{3mm}
\begin{centering}
%{\smaller
\begin{tabular}[b]{|r||c|c|c|}
\hline
Net & ULP & SFA & algDNC\tabularnewline
\hline 
\hline
20 & 0:00:12 & 0:00:07 & 0:00:13\tabularnewline
\hline
40 & 0:05:49 & 0:00:30 & 0:00:16\tabularnewline
\hline
60 & 0:39:44 & 0:01:49 & 0:00:48\tabularnewline
\hline
80 & 5:00:33 & 0:03:01 & 0:01:34\tabularnewline
\hline
100 & 5:22:02 & 0:03:39 & 0:02:31\tabularnewline
\hline
120 & 20:45:44 & 0:10:19 & 0:03:41\tabularnewline
\hline
140 & 33:15:36 & 0:11:18 & 0:05:58\tabularnewline
\hline
160 & 58:06:08 & 0:10:55 & 0:07:05\tabularnewline
\hline
180 & $\sim$13 days & 0:29:52 & 0:10:27\tabularnewline
\hline
200 & -- & 0:12:15 & 0:10:13\tabularnewline
\hline
220 & -- & 0:51:16 & 0:18:23\tabularnewline
\hline
\multicolumn{4}{|c|}{$\cdots$}\tabularnewline
\hline
400 & -- & 128:27:16 & 1:47:39\tabularnewline
\hline
\multicolumn{4}{|c|}{$\cdots$}\tabularnewline
\hline
$10^3$ & -- & -- & $\;$13:45:52$\;$\tabularnewline
\hline
\multicolumn{4}{c}{
%\footnotesize
Network analysis time (hh:mm:ss)}\tabularnewline
\multicolumn{4}{c}{
%\footnotesize
of single threaded analysis runs.}\tabularnewline
\end{tabular}
%}
\par\end{centering}
\caption{\label{fig:AnalysisEffort}Computational Effort.\emph{ }The ULP becomes
computationally infeasible at moderate network sizes, its analysis
time increases to $\sim\negmedspace13$~days at $180$ devices. The
old SFA (without improvements) scales better, yet, its cost becomes
unpredictable when increasing the network size. Moreover, it reaches
$>\negmedspace5$~days at $400$ devices. Our new algDNC scales better
and is more resilient to the randomness of our network creation. Its
underlying improvements can benefit the inaccurate SFA similarly.}
\end{figure*}

\section{Related Work}

\label{sec:RelatedWork}In this work, we rely on the PMOO left-over
service curve under arbitrary multiplexing of flows. DNC also offers
an algebraically derived left-over service curve for FIFO-multiplexing
servers $\beta_{\theta}^{\text{l.o.}}$ \cite{LeBoudec_NCbook}. Similar
to the arbitrary-multiplexing one, it is only applicable to a single
server. Thus, it allows for an analysis akin to the SFA. Again, this
causes the different problems, both on tandems where multiplexing
with cross-flows is paid for more than once (Section~\ref{subsec:Network-Calculus-Analysis})
and in the compositional feed-forward network analysis (Section~\ref{subsec:Problems-of-algDNC}).

Effort to achieve the PMOO principle in the analysis of FIFO-multiplexing
servers resulted in the Least Upper Delay Bound (LUDB) analysis~\cite{Lenzini_methodologyFIFOe2e,LUDB_journal}.
If paths of cross-flows do not overlap (i.e., the flows are nested
into each other), the LUDB suggests to convolve servers before removing
cross-flows, subject to the nesting. The latter step is done by computing
the above FIFO left-over service curve. E.g., in Figure~\ref{fig:Sample-server-topology},
$x\negmedspace f_{1}$'s arrival at server $s_{1}$ ($\alpha_{s_{1}}^{x\negmedspace f_{1}}$)
is removed from $s_{1}$, then the left-over service curve $\beta_{\theta,s_{1}}^{\text{l.o.}}$
is convolved with $\beta_{s_{2}}$, and finally $\alpha_{s_{2}}^{x\negmedspace f_{1}}$
is removed from this curve. Note, that this approach for tandem analysis
enforces a PSOO violation during cross-traffic arrival bounding like
any other algDNC tandem analysis. If paths of cross-flows are not
nested, this approach cannot be applied, for instance, if we analyzed
$x\!f_{2}$ in Figure~\ref{fig:Sample-server-topology}. In this
case, \cite{LUDB_journal} suggests to decompose into sub-tandems
with nested interference patterns before applying the LUDB procedure.
The end-to-end left-over service curve for delay bounding is then
derived by convolution of the sub-tandem ones. Oftentimes, there are
multiple alternatives to decompose to nested interference patterns
(Figure~\ref{fig:Sample-server-topology}: both links are potential
demarcations of a decomposition). LUDB suggests exhaustive enumeration
of alternatives, computation of all delay bounds and returning the
least one among them~\textendash{} the compFFA procedure depicted
in Section~\ref{subsec:alg_comp_DNC} is strictly executed. Obvious
drawbacks of this procedure are a lack of the PMOO principle implementation
and a combinatorial explosion as shown in Section~\ref{subsec:ExhaustiveDerivation123}.

The LUDB's sub-tandem decomposition and our novel algDNC's decomposition
are similar, yet, they differ in some key aspects. First, our algDNC
does not require to result in tandems with nested interference as
the arbitrary-multiplexing PMOO analysis fully implements the eponymous
principle. Our exhaustive approach thus results in more decompositions
per tandem than the LUDB. Therefore, more permissible (min,+)-equations
are derived and solved. Secondly, \cite{Lenzini_methodologyFIFOe2e}
and \cite{LUDB_journal} are concerned with a tandem analysis only.
They do not address the problems in feed-forward analyses presented
in Section~\ref{subsec:Problems-of-algDNC}, although the LUDB enforces
PSOO violations. Moreover, they do not provide a technical solution
for the potential combinatorial explosion problem; \cite{LUDB_journal}
rather presents a heuristic to trade accuracy against computational
effort on a single tandem. Our novel analysis design is very generic
and thus not restricted to arbitrary multiplexing. The LUDB tandem
analysis can be embedded into it and benefit from our efficiency improvements.

An optimization-based DNC approach for tight FIFO-multiplexing feed-forward
network analysis exists as well~\cite{BouillardStea_FF_FIFO}. It
transforms the DNC description of the network to a Mixed-Integer Linear
Program (MILP) where the integer variables encode the (partially)
parallel paths of flows. This circumvents the step of explicitly extending
a partial order to the set of all compatible total orders, the root
cause the LP analysis' combinatorial explosion. Again, the computational
effort to solve the MILP for large networks is not evaluated. Instead,
the authors advise to remove constraints such that all integer variables
are removed, leaving an ordinary linear program to solve. I.e., tightness
is traded for computational effort reduction. This is similar to the
ULP heuristic that we showed to be computationally infeasible nonetheless.

\section{Conclusion}

\label{sec:Conclusion}In this article, we contribute the first DNC
analysis for high quality and low cost end-to-end delay bounds in
large feed-forward networks. We demonstrate this contribution on the
novel set of feed-forward networks we created for DNC evaluation as
well as the first comprehensive optDNC evaluation in such networks.
Figure~\ref{fig:DNC-QandC-after} summarizes the findings of our
article: Against previous belief, we showed that optDNC's most efficient
heuristic, ULP, is computationally infeasible even for moderately
sized networks. For larger networks, we also showed that the algebraic
SFA is more costly than expected, becoming barely feasible to execute.
The remaining, grayed-out TFA and PMOO were not categorized more precisely~\textendash{}
their quality does not change and more detailed knowledge about their
cost would still have left the desired area in the intersection of
low cost and high quality empty. Therefore, we developed a novel algDNC
analysis that combines the strengths of all previously existing analyses,
crucially improves their quality and incorporates decisive efficiency
improvements. This novel algDNC analysis is currently the only network
analysis providing highly accurate delay bounds while being computationally
feasible even in large feed-forward networks. Its algebraically derived
delay bounds deviate from the optimization-based ULP analysis by only
1.142\% on average in our evaluations while computation times are
several orders of magnitude smaller.
\begin{figure}[H]
\begin{centering}
\includegraphics[width=0.55\columnwidth]{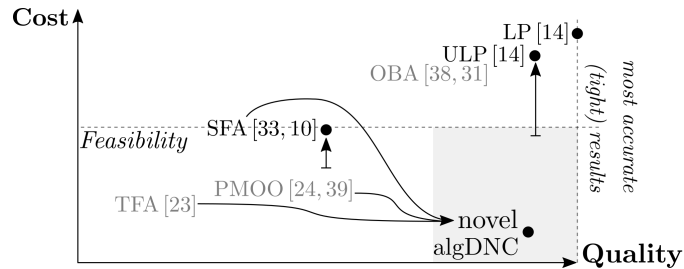}
\par\end{centering}
\centering{}\caption{\textbf{\label{fig:DNC-QandC-after}}Categorization of DNC analyses
regarding their quality and cost. For analyses with black dots we
generated sufficient evaluation data to pinpoint them with high certainty.
Our novel algDNC analysis is evidently the only network analysis of
good quality that is also feasible to execute.}
\end{figure}

\bibliographystyle{abbrv}
\bibliography{CR_00_ArXiv}

\appendix

\section{Generation of Networks for DNC Scalability Tests}

\label{sec:GenerateTestNetworks}

\subsection{Device Graph and Server Graph}

Data communication networks are commonly modeled as graphs where nodes
represent individual devices like a router or a switch. These devices
can have multiple outputs to connect to other devices. Network calculus
analyzes flows that cross the servers at the output of devices. Therefore,
it needs to transform the device graph representation of a network
into a server graph representation~\cite{BS10-1,BS15} where these
servers are directly connected. We use the term \emph{network} to
refer to a server graph crossed by flows.

\subsection{Network Generation}

\label{subsec:SetNetworks}Currently, there is neither standard set
of feed-forward networks to test a DNC analysis on nor a common procedure
to create them. Previous work provided small networks tailored to
illustrate a specific advancement~\cite{Jens_ArbMuxLurch,BouillardINFOCOM2010,BS16-1}
and oftentimes restricted its tool support to these networks. In contrast,
we created an extensible set of networks as follows:

We start with a device graph provided by a topology generator (aSHIIP~\cite{aSHIIP}).
We used the general linear preference (GLP) model~\cite{Towsley_GLP}
with its provided default parameter setting ($m_{0}=20$, $m=1$,
$p=0.4695$, $\beta_{\text{GLP}}=0.6447$) to create Internet-like
topologies of different sizes. Then, we transformed them to their
server graph representation where servers resemble the capability
to transmit via $10$Gbps links (service curves $\beta_{10\text{Gbps},0}$).
Next, we applied the turn prohibition algorithm~\cite{NC_TurnProhibition}
to break potential cycles (``feed-forwardize the network'') and
added flows with a fixed server-to-flow ratio of 1:4 for all network
sizes. Flows' arrivals were uniformly shaped to token buckets with
rate $5$Mbps and bucket size $5$Mb, i.e., $\gamma_{5\text{Mbps},5\text{Mb}}$,
to scale up from the unit-sized arrivals used in the AFDX evaluation.
For path creation, a pair of source/sink devices were randomly selected
from the device graph. The shortest path between these devices, yet,
in the feed-forwardized server graph, was then computed. Table~\ref{tab:NetServerFlow}
show the resulting network sized we evaluate in Sections~\ref{fig:Eval-AFDX-DelayBounds}
and~\ref{fig:Eval-AFDX-RunTimes}.
\begin{table}[H]
\centering{}\caption{\label{tab:NetServerFlow}Networks to Evaluate Quality and Cost of
DNC.}
\begin{tabular}{|c|c|c|}
\hline 
Devices & Servers & Flows\tabularnewline
\hline 
\hline 
20 & 38 & 152\tabularnewline
\hline 
40 & 118 & 472\tabularnewline
\hline 
60 & 164 & 656\tabularnewline
\hline 
80 & 282 & 1128\tabularnewline
\hline 
100 & 364 & 1456\tabularnewline
\hline 
120 & 398 & 1592\tabularnewline
\hline 
140 & 512 & 2048\tabularnewline
\hline 
160 & 572 & 2288\tabularnewline
\hline 
180 & 646 & 2584\tabularnewline
\hline 
\end{tabular}\hspace{5mm}%
\begin{tabular}{|c|c|c|}
\hline 
Devices & Servers & Flows\tabularnewline
\hline 
\hline 
200 & 740 & 2960\tabularnewline
\hline 
220 & 744 & 2976\tabularnewline
\hline 
240 & 882 & 3528\tabularnewline
\hline 
260 & 976 & 3904\tabularnewline
\hline 
280 & 994 & 3976\tabularnewline
\hline 
300 & 1124 & 4496\tabularnewline
\hline 
400 & 1478 & 5912\tabularnewline
\hline 
500 & 1876 & 7504\tabularnewline
\hline 
1000 & 3626 & 14504\tabularnewline
\hline 
\end{tabular}
\end{table}

\section{(min,+)-Equations for Figures~4 and 5}

\label{sec:Appendix-Quality-algDNC}We show that the previously neglected
decomposition alternative for an algDNC analysis, Alt3 shown in Figure~\ref{fig:Fig4Alt3},
can arbitrarily outperform the existing analyses SFA and PMOO. First,
we derive the respective left-over service curves, $\beta_{\left\langle s_{1},s_{2}\right\rangle }^{\text{l.o.Alt3}}$,
$\beta_{\left\langle s_{1},s_{2}\right\rangle }^{\text{l.o.SFA}}$,
and $\beta_{\left\langle s_{1},s_{2}\right\rangle }^{\text{l.o.PMOO}}$,
for the flow of interest (foi) in Figure~\ref{fig:Sample-server-topology}.
They are crucial for each (min,+)-equation bounding the foi's delay.
\begin{eqnarray*}
\beta_{\left\langle s_{1},s_{2}\right\rangle }^{\text{l.o.Alt3}} & = & \beta_{\left\langle s_{1},s_{2}\right\rangle }\ominus\left(\alpha_{s_{1}}^{\left[x\!f_{1},x\!f_{2}\right]},\alpha_{s_{2}}^{x\!f_{2}}\right)\\
 & = & \beta_{s_{1}}^{\text{l.o.foi}}\otimes\beta_{s_{2}}^{\text{l.o.foi}}\\
 & = & \left(\beta_{s_{1}}\ominus\alpha_{s_{1}}^{\left[x\!f_{1},x\!f_{2}\right]}\right)\otimes\left(\beta_{s_{2}}\ominus\alpha_{s_{2}}^{x\!f_{2}}\right)\\
 & = & \left(\beta_{s_{1}}\ominus\left(\alpha_{s_{0}}^{\left[x\!f_{1},x\!f_{2}\right]}\oslash\beta_{s_{0}}\right)\right)\otimes\left(\beta_{s_{2}}\ominus\left(\alpha_{s_{0}}^{x\!f_{1}}\oslash\beta_{\left\langle s_{0},s_{1}\right\rangle }^{\text{l.o.}x\!f_{2}}\right)\right)\\
 & = & \left(\beta_{s_{1}}\ominus\left(\alpha^{\left[x\!f_{1},x\!f_{2}\right]}\oslash\beta_{s_{0}}\right)\right)\otimes\left(\beta_{s_{2}}\ominus\left(\alpha^{x\!f_{1}}\oslash\left(\beta_{\left\langle s_{0},s_{1}\right\rangle }\ominus\left(\alpha^{x\!f_{2}}\right)\right)\right)\right)\\
\\
\beta_{\left\langle s_{1},s_{2}\right\rangle }^{\text{l.o.SFA}} & = & \beta_{\left\langle s_{1},s_{2}\right\rangle }\ominus\left(\alpha_{s_{1}}^{\left[x\!f_{1},x\!f_{2}\right]},\alpha_{s_{2}}^{x\!f_{2}}\right)\\
 & = & \beta_{s_{1}}^{\text{l.o.foi}}\otimes\beta_{s_{2}}^{\text{l.o.foi}}\\
 & = & \left(\beta_{s_{1}}\ominus\alpha_{s_{1}}^{\left[x\!f_{1},x\!f_{2}\right]}\right)\otimes\left(\beta_{s_{2}}\ominus\alpha_{s_{2}}^{x\!f_{2}}\right)\\
 & = & \left(\beta_{s_{1}}\ominus\left(\alpha_{s_{0}}^{\left[x\!f_{1},x\!f_{2}\right]}\oslash\beta_{s_{0}}^{\text{l.o.}\left[x\!f_{1},x\!f_{2}\right]}\right)\right)\otimes\left(\beta_{s_{2}}\ominus\left(\alpha_{s_{2}}^{x\!f_{2}}\oslash\beta_{\left\langle s_{0},s_{1}\right\rangle }^{\text{l.o.}x\!f_{2}}\right)\right)\\
 & = & \left(\beta_{s_{1}}\ominus\left(\left(\alpha_{s_{0}}^{x\!f_{1}}+\alpha_{s_{0}}^{x\!f_{2}}\right)\oslash\beta_{s_{0}}\right)\right)\otimes\left(\beta_{s_{2}}\ominus\left(\alpha^{x\!f_{2}}\oslash\left(\beta_{s_{0}}^{\text{l.o.}x\!f_{2}}\otimes\beta_{s_{1}}^{\text{l.o.}x\!f_{2}}\right)\right)\right)\\
 & = & \left(\beta_{s_{1}}\ominus\left(\left(\alpha^{x\!f_{1}}+\alpha^{x\!f_{2}}\right)\oslash\beta_{s_{0}}\right)\right)\otimes\left(\beta_{s_{2}}\ominus\left(\alpha^{x\!f_{2}}\oslash\left(\left(\beta_{s_{0}}\ominus\alpha_{s_{0}}^{x\!f_{1}}\right)\otimes\left(\beta_{s_{1}}\ominus\alpha_{s_{1}}^{x\!f_{1}}\right)\right)\right)\right)\\
 & = & \left(\beta_{s_{1}}\ominus\left(\left(\alpha^{x\!f_{1}}+\alpha^{x\!f_{2}}\right)\oslash\beta_{s_{0}}\right)\right)\otimes\left(\beta_{s_{2}}\ominus\left(\alpha^{x\!f_{2}}\oslash\left(\left(\beta_{s_{0}}\ominus\alpha^{x\!f_{1}}\right)\otimes\left(\beta_{s_{1}}\ominus\left(\alpha^{x\!f_{1}}\oslash\beta_{s_{0}}^{\text{l.o.}x\!f_{1}}\right)\right)\right)\right)\right)\\
 & = & \left(\beta_{s_{1}}\ominus\left(\left(\alpha^{x\!f_{1}}+\alpha^{x\!f_{2}}\right)\oslash\beta_{s_{0}}\right)\right)\otimes\left(\beta_{s_{2}}\ominus\left(\alpha^{x\!f_{2}}\oslash\left(\left(\beta_{s_{0}}\ominus\alpha^{x\!f_{1}}\right)\otimes\left(\beta_{s_{1}}\ominus\left(\alpha^{x\!f_{1}}\oslash\left(\beta_{s_{0}}\ominus\alpha^{x\!f_{2}}\right)\right)\right)\right)\right)\right)\\
\\
\beta_{\left\langle s_{1},s_{2}\right\rangle }^{\text{l.o.PMOO}} & = & \beta_{\left\langle s_{1},s_{2}\right\rangle }\ominus\left(\alpha_{s_{1}}^{x\!f_{1}},\alpha_{s_{1}}^{x\!f_{2}}\right)\\
 & = & \beta_{\left\langle s_{1},s_{2}\right\rangle }\ominus\left(\alpha^{x\!f_{1}}\oslash\beta_{s_{0}}^{\text{l.o.}x\!f_{1}},\alpha^{x\!f_{2}}\oslash\beta_{s_{0}}^{\text{l.o.}x\!f_{2}}\right)\\
 & = & \beta_{\left\langle s_{1},s_{2}\right\rangle }\ominus\left(\alpha^{x\!f_{1}}\oslash\left(\beta_{s_{0}}\ominus\alpha^{x\!f_{2}}\right),\alpha^{x\!f_{2}}\oslash\left(\beta_{s_{0}}\ominus\alpha^{x\!f_{1}}\right)\right)
\end{eqnarray*}
Next, we construct a sample parameter setting for the curves in our
three (min,+)-equations that simplifies the complex tandem left-over
derivations $\beta_{\mathcal{T}}^{\text{l.o.}}$. This parameter setting
allows us to continue by arguing over the exact curve shapes.

Assume the following arrival curves from $\mathcal{F}_{\text{TB}}$:
$\alpha^{x\!f_{1}}(d)=r^{x\!f_{1}}\cdot d$, $\alpha^{x\!f_{2}}(d)=r^{x\!f_{2}}\cdot d$,
and $\alpha^{\text{foi}}(d)=r^{\text{foi}}\cdot d$. Moreover assume
the following strict service curves from $\mathcal{F}_{\text{RL}}$:
$\beta_{s_{n}}(d)=R_{s_{n}}$, $n\in\left\{ 1,2\right\} $, $\beta_{s_{0}}(d)=R_{s_{0}}\cdot\max\left\{ 0,\,d-T_{s_{0}}\right\} $,
with $R_{s_{m}}\ge r^{x\!f_{1}}+r^{x\!f_{2}}+r^{\text{foi}}$, $m\in\left\{ 0,1,2\right\} $
for finite delay bounds. For positive burstiness increase of the cross-flows,
assume $T_{s_{0}}>0$, i.e., the burst terms (bucket sizes of $\mathcal{F}_{\text{TB}}$)
$b_{s_{n}}^{\text{AB}\mathbb{F}}$ become positive for each arrival
bounding $\text{AB}\in\left\{ \text{SFA},\text{PMOO,Alt3}\right\} $
of any cross-flow (aggregate) $\mathbb{F}\in\left\{ x\!f_{1},\,x\!f_{2},\,\left[x\!f_{1},x\!f_{2}\right]\right\} $,
at servers $s_{1}$ and $s_{2}$.

In this simplified setting, the flow of interest's delay bound equals
the $\beta_{\left\langle s_{1},s_{2}\right\rangle }^{\text{l.o.}}$'s
latency term $T_{\left\langle s_{1},s_{2}\right\rangle }^{\text{l.o.}}$.
Therefore, we derive the latency terms for all three decomposition.
The SFA left-over latency, abbreviated $T_{\left\langle s_{1},s_{2}\right\rangle }^{\text{l.o.SFA}}$
, respectively the SFA delay $D^{\text{SFA}}$, are
\[
D^{\text{SFA}}\,=\,T_{\left\langle s_{1},s_{2}\right\rangle }^{\text{l.o.SFA}}\,=\,T_{s_{1}}^{\text{l.o.SFA}}+T_{s_{2}}^{\text{l.o.SFA}}\,=\,\frac{b_{s_{1}}^{\text{SFA}x\!f_{1}}+b_{s_{1}}^{\text{SFA}x\!f_{2}}}{R_{s_{1}}\!-r^{x\!f_{1}}-r^{x\!f_{2}}}+\frac{b_{s_{2}}^{\text{SFA}x\!f_{2}}}{R_{s_{2}}\!-r^{x\!f_{2}}},
\]
the PMOO left-over latency $T_{\left\langle s_{1},s_{2}\right\rangle }^{\text{l.o.PMOO}}$
and delay $D^{\text{PMOO}}$ are
\[
T_{\left\langle s_{1},s_{2}\right\rangle }^{\text{l.o. PMOO}}\,=\,D^{\text{PMOO}}\,=\,\frac{b_{s_{1}}^{\text{PMOO}x\!f_{1}}+b_{s_{1}}^{\text{PMOO}x\!f_{2}}}{\left(R_{s_{1}}\!-r^{x\!f_{1}}-r^{x\!f_{2}}\right)\wedge\left(R_{s_{2}}\!-r^{x\!f_{2}}\right)},
\]
and Alt3's left-over latency $T_{\left\langle s_{1},s_{2}\right\rangle }^{\text{l.o.Alt3}}$
and delay $D^{\text{Alt3}}$ are
\[
D^{\text{Alt3}}\,=\,T^{\text{l.o.Alt3}}\,=\,T_{s_{1}}^{\text{l.o.Alt3}}+T_{s_{2}}^{\text{l.o.Alt3}}\,=\,\frac{b_{s_{1}}^{\text{Alt3}\left[x\!f_{1},x\!f_{2}\right]}}{R_{s_{1}}\!-r^{x\!f_{1}}-r^{x\!f_{2}}}+\frac{b_{s_{2}}^{\text{Alt3}x\!f_{2}}}{R_{s_{2}}\!-r^{x\!f_{2}}}.
\]
We see that $D^{\text{PMOO}}\le D^{\text{SFA}}$, iff
\begin{eqnarray*}
\frac{b_{s_{1}}^{\text{PMOO}x\!f_{1}}+b_{s_{1}}^{\text{PMOO}x\!f_{2}}}{\left(R_{s_{1}}\!-r^{x\!f_{1}}-r^{x\!f_{2}}\right)\wedge\left(R_{s_{2}}\!-r^{x\!f_{2}}\right)} & \le & \frac{b_{s_{1}}^{\text{SFA}x\!f_{1}}+b_{s_{1}}^{\text{SFA}x\!f_{2}}}{R_{s_{1}}\!-r^{x\!f_{1}}-r^{x\!f_{2}}}+\frac{b_{s_{2}}^{\text{SFA}x\!f_{2}}}{R_{s_{2}}\!-r^{x\!f_{2}}}
\end{eqnarray*}
1)\emph{ }$R_{s_{1}}\!-\!r^{x\!f_{1}}\!-\!r^{x\!f_{2}}\!>\!R_{s_{2}}\!-\!r^{x\!f_{2}}$:
\begin{eqnarray*}
\frac{b_{s_{1}}^{\text{PMOO}x\!f_{1}}+b_{s_{1}}^{\text{PMOO}x\!f_{2}}}{R_{s_{2}}\!-r^{x\!f_{2}}} & \le & \frac{b_{s_{1}}^{\text{SFA}x\!f_{1}}+b_{s_{1}}^{\text{SFA}x\!f_{2}}}{R_{s_{1}}\!-r^{x\!f_{1}}-r^{x\!f_{2}}}+\frac{b_{s_{2}}^{\text{SFA}x\!f_{2}}}{R_{s_{2}}\!-r^{x\!f_{2}}}\\
\frac{b_{s_{1}}^{\text{PMOO}x\!f_{1}}+b_{s_{1}}^{\text{PMOO}x\!f_{2}}-b_{s_{2}}^{\text{SFA}x\!f_{2}}}{R_{s_{2}}\!-r^{x\!f_{2}}} & \le & \frac{b_{s_{1}}^{\text{SFA}x\!f_{1}}+b_{s_{1}}^{\text{SFA}x\!f_{2}}}{R_{s_{1}}\!-r^{x\!f_{1}}-r^{x\!f_{2}}}\\
\frac{R_{s_{1}}\!-r^{x\!f_{1}}-r^{x\!f_{2}}}{R_{s_{2}}\!-r^{x\!f_{2}}} & \le & \frac{b_{s_{1}}^{\text{SFA}x\!f_{1}}+b_{s_{1}}^{\text{SFA}x\!f_{2}}}{b_{s_{1}}^{\text{PMOO}x\!f_{1}}+b_{s_{1}}^{\text{PMOO}x\!f_{2}}-b_{s_{2}}^{\text{SFA}x\!f_{2}}}
\end{eqnarray*}
2) $R_{s_{2}}\!-r^{x\!f_{2}}\ge R_{s_{1}}\!-r^{x\!f_{1}}-r^{x\!f_{2}}$:
\begin{eqnarray*}
\frac{b_{s_{1}}^{\text{PMOO}x\!f_{1}}+b_{s_{1}}^{\text{PMOO}x\!f_{2}}}{R_{s_{1}}\!-r^{x\!f_{1}}-r^{x\!f_{2}}} & \le & \frac{b_{s_{1}}^{\text{SFA}x\!f_{1}}+b_{s_{1}}^{\text{SFA}x\!f_{2}}}{R_{s_{1}}\!-r^{x\!f_{1}}-r^{x\!f_{2}}}+\frac{b_{s_{2}}^{\text{SFA}x\!f_{2}}}{R_{s_{2}}\!-r^{x\!f_{2}}}\\
\frac{b_{s_{1}}^{\text{PMOO}x\!f_{1}}+b_{s_{1}}^{\text{PMOO}x\!f_{2}}-b_{s_{1}}^{\text{SFA}x\!f_{1}}+b_{s_{1}}^{\text{SFA}x\!f_{2}}}{R_{s_{1}}\!-r^{x\!f_{1}}-r^{x\!f_{2}}} & \leq & \frac{b_{s_{2}}^{\text{SFA}x\!f_{2}}}{R_{s_{2}}\!-r^{x\!f_{2}}}\\
\frac{R_{s_{2}}\!-\!r^{x\!f_{2}}}{R_{s_{1}}\!-\!r^{x\!f_{1}}\!-\!r^{x\!f_{2}}} & \leq & \frac{b_{s_{2}}^{\text{SFA}x\!f_{2}}}{b_{s_{1}}^{\text{PMOO}x\!f_{1}}\!-\!b_{s_{1}}^{\text{SFA}x\!f_{1}}\!+\!b_{s_{1}}^{\text{PMOO}x\!f_{2}}\!-\!b_{s_{1}}^{\text{SFA}x\!f_{2}}}
\end{eqnarray*}
Neither condition is strictly fulfilled as both depend on the service
rates at $s_{1}$ and $s_{2}$. In \cite{Jens_ArbMuxLurch}, it was
shown that the PMOO \emph{tandem} analysis cannot exploit fast residual
service rates at the end; this derivation shows that the PMOO \emph{network}
analysis also struggles with fast rates at the front of the flow of
interest's path due the enforced PSOO violation. 

Last, we show that the previously neglected Alt3 can outperform the
existing compositional feed-forward analyses SFA and PMOO. The relations
between $D^{\text{Alt3}}$, $D^{\text{SFA}}$ and $D^{\text{PMOO}}$
are:\\
1) $D^{\text{Alt3}}\le D^{\text{SFA}}$, iff
\begin{eqnarray*}
\frac{b_{s_{1}}^{\text{Alt3}\left[x\!f_{1},x\!f_{2}\right]}}{R_{s_{1}}\!-r^{x\!f_{1}}-r^{x\!f_{2}}}+\frac{b_{s_{2}}^{\text{Alt3}x\!f_{2}}}{R_{s_{2}}\!-r^{x\!f_{2}}} & \le & \frac{b_{s_{1}}^{\text{SFA}x\!f_{1}}+b_{s_{1}}^{\text{SFA}x\!f_{2}}}{R_{s_{1}}\!-r^{x\!f_{1}}-r^{x\!f_{2}}}+\frac{b_{s_{2}}^{\text{SFA}x\!f_{2}}}{R_{s_{2}}\!-r^{x\!f_{2}}}\\
\frac{b_{s_{1}}^{\text{Alt3}\left[x\!f_{1},x\!f_{2}\right]}-b_{s_{1}}^{\text{SFA}x\!f_{1}}-b_{s_{1}}^{\text{SFA}x\!f_{2}}}{R_{s_{1}}\!-r^{x\!f_{1}}-r^{x\!f_{2}}} & \le & \frac{b_{s_{2}}^{\text{SFA}x\!f_{2}}-b_{s_{2}}^{\text{Alt3}x\!f_{2}}}{R_{s_{2}}\!-r^{x\!f_{2}}}\\
\frac{R_{s_{2}}\!-r^{x\!f_{2}}}{R_{s_{1}}\!-r^{x\!f_{1}}-r^{x\!f_{2}}} & \le & \frac{b_{s_{2}}^{\text{SFA}x\!f_{2}}-b_{s_{2}}^{\text{Alt3}x\!f_{2}}}{b_{s_{1}}^{\text{Alt3}\left[x\!f_{1},x\!f_{2}\right]}-b_{s_{1}}^{\text{SFA}x\!f_{1}}-b_{s_{1}}^{\text{SFA}x\!f_{2}}}
\end{eqnarray*}
2)\emph{ }$D^{\text{Alt3}}\le D^{\text{PMOO}}$, iff\\
\tabto{5mm}a) $R_{s_{1}}\!-r^{x\!f_{1}}\!-r^{x\!f_{2}}>R_{s_{2}}\!-r^{x\!f_{2}}$:
\begin{eqnarray*}
\frac{b_{s_{1}}^{\text{Alt3}\left[x\!f_{1},x\!f_{2}\right]}}{R_{s_{1}}\!-r^{x\!f_{1}}-r^{x\!f_{2}}}+\frac{b_{s_{2}}^{\text{Alt3}x\!f_{2}}}{R_{s_{2}}\!-r^{x\!f_{2}}} & \le & \frac{b_{s_{1}}^{\text{PMOO}x\!f_{1}}+b_{s_{1}}^{\text{PMOO}x\!f_{2}}}{R_{s_{2}}\!-r^{x\!f_{2}}}\\
\frac{b_{s_{1}}^{\text{Alt3}\left[x\!f_{1},x\!f_{2}\right]}}{R_{s_{1}}\!-r^{x\!f_{1}}-r^{x\!f_{2}}} & \leq & \frac{b_{s_{1}}^{\text{PMOO}x\!f_{1}}+b_{s_{1}}^{\text{PMOO}x\!f_{2}}-b_{s_{2}}^{\text{Alt3}x\!f_{2}}}{R_{s_{2}}\!-r^{x\!f_{2}}}\\
\frac{R_{s_{2}}\!-r^{x\!f_{2}}}{R_{s_{1}}\!-r^{x\!f_{1}}-r^{x\!f_{2}}} & \leq & \frac{b_{s_{1}}^{\text{PMOO}x\!f_{1}}+b_{s_{1}}^{\text{PMOO}x\!f_{2}}-b_{s_{2}}^{\text{Alt3}x\!f_{2}}}{b_{s_{1}}^{\text{Alt3}\left\{ x\!f_{1},x\!f_{2}\right\} }}
\end{eqnarray*}
\tabto{5mm}b) $R_{s_{2}}\!-r^{x\!f_{2}}\ge R_{s_{1}}\!-r^{x\!f_{1}}\!-r^{x\!f_{2}}$:
\begin{eqnarray*}
\frac{b_{s_{1}}^{\text{Alt3}\left[x\!f_{1},x\!f_{2}\right]}}{R_{s_{1}}\!-r^{x\!f_{1}}-r^{x\!f_{2}}}+\frac{b_{s_{2}}^{\text{Alt3}x\!f_{2}}}{R_{s_{2}}\!-r^{x\!f_{2}}} & \le & \frac{b_{s_{1}}^{\text{PMOO}x\!f_{1}}+b_{s_{1}}^{\text{PMOO}x\!f_{2}}}{R_{s_{1}}\!-r^{x\!f_{1}}-r^{x\!f_{2}}}\\
\frac{b_{s_{2}}^{\text{Alt3}x\!f_{2}}}{R_{s_{2}}\!-r^{x\!f_{2}}} & \leq & \frac{b_{s_{1}}^{\text{PMOO}x\!f_{1}}+b_{s_{1}}^{\text{PMOO}x\!f_{2}}-b_{s_{1}}^{\text{Alt3}\left[x\!f_{1},x\!f_{2}\right]}}{R_{s_{1}}\!-r^{x\!f_{1}}-r^{x\!f_{2}}}\\
\frac{R_{s_{1}}\!-r^{x\!f_{1}}-r^{x\!f_{2}}}{R_{s_{2}}\!-r^{x\!f_{2}}} & \leq & \frac{b_{s_{1}}^{\text{PMOO}x\!f_{1}}+b_{s_{1}}^{\text{PMOO}x\!f_{2}}-b_{s_{1}}^{\text{Alt3}\left[x\!f_{1},x\!f_{2}\right]}}{b_{s_{2}}^{\text{Alt3}x\!f_{2}}}
\end{eqnarray*}
The observations about the relation of delay bounds still hold: Relations
1), 2a), and 2b) reflect the influence of the rates on the flow of
interest's path (left terms). A large service rate $R_{s_{1}}$ can,
in fact, best be exploited by Alt3. Alt3 can thus simultaneously outperform
both existing analyses. However, there is no strict ordering with
respect to delay bounds; each analysis can potentially outperform
the others.

Considering the parameters omitted for the ease of presentation, we
can find a simple parameter setting that allows Alt3 to arbitrarily
outperform SFA and PMOO. Alt3 scales better with increasing $b^{x\negmedspace f_{1}}$
when parameters are set to: $\beta_{s_{0}}=\beta_{25,5}$, $\beta_{s_{1}}=\beta_{25,0}$,
$\beta_{s_{2}}=\beta_{3,5}$, $\alpha^{\text{foi}}=\gamma_{0.5,5}$,
$\alpha^{xf_{1}}=\gamma_{2.5,b^{xf_{1}}}$, $\alpha^{xf_{2}}=\gamma_{2.5,5}$.
This is illustrated by the following Figure.

\begin{figure}[H]
\begin{centering}
\includegraphics[width=0.6\columnwidth]{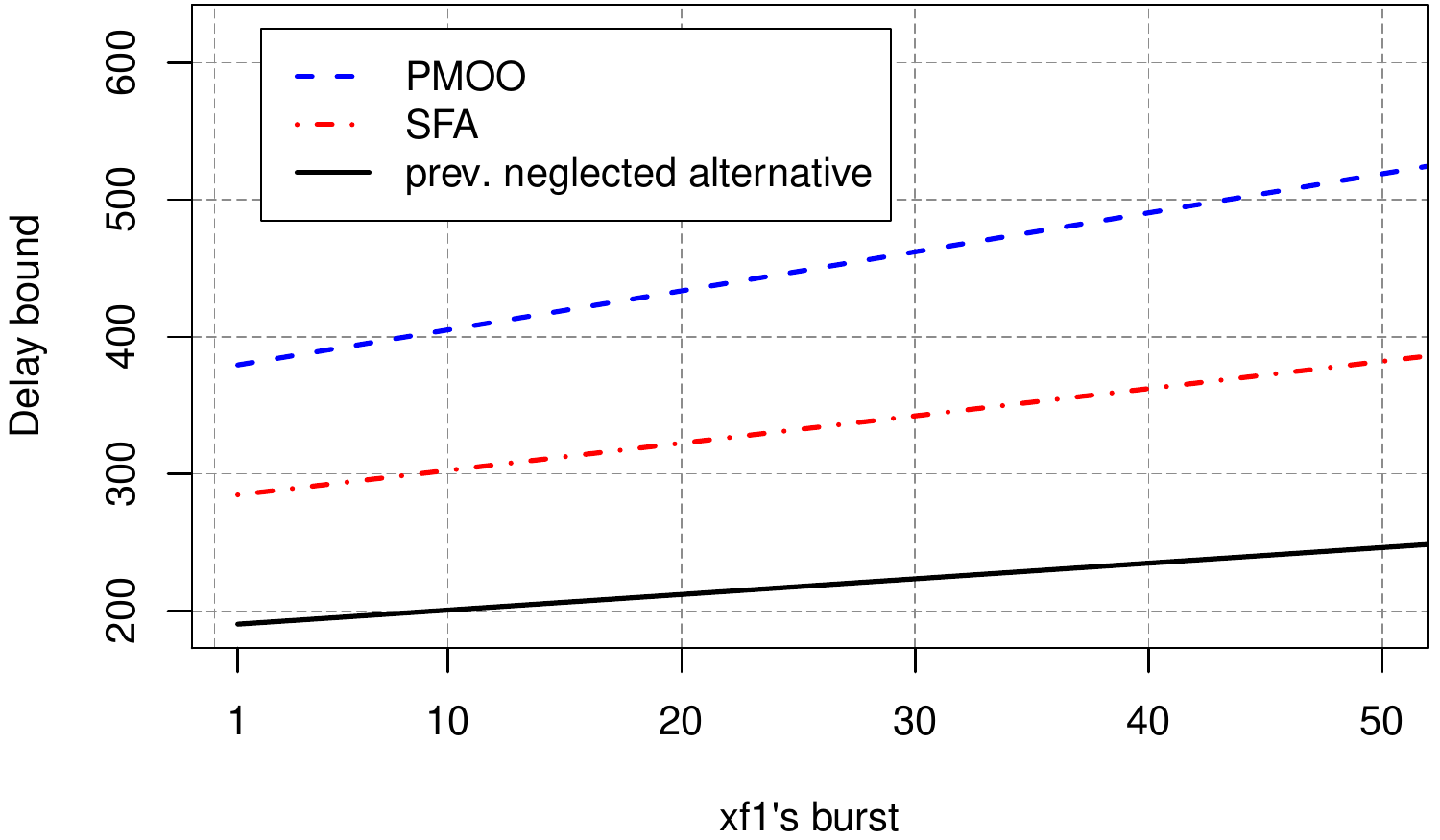}
\par\end{centering}
\caption{\label{fig:TMnew_arbBetter}The previously neglected alternative (see
Section~\ref{subsec:ProblemDesign-compFFA}) can arbitrarily outperform
PMOO and SFA in the network shown in of Figure~\ref{fig:Sample-server-topology}. }
\end{figure}

\section{Computational Complexity}

\label{sec:Computational-complexity-algDNC}The computational complexity
of algDNC algorithms mainly depends on two aspects:
\begin{enumerate}
\item the shape of curves used in the model (Section~\ref{subsec:System-Description})
defines the complexity of the operations applied to them (Section~\ref{subsec:Network-Calculus-Analysis}).
\item the network put into the algorithm, including the entanglement of
flows crossing it, defines the operations in a permissible algDNC
equation (Section~\ref{subsub:Server-Graph}).
\end{enumerate}
For (1), complexity of basic operations,~\cite{Bouillard_JDEDS}
gives results. Most notably, if arrival curves are from the set $\mathcal{F}_{\text{TB}}$
and service curves are from the set $\mathcal{F}_{\text{RL}}$, the
(min,+)-algebraic DNC operations are in $\mathcal{O}(1)$. The second
impact factor, the analyzed network, has not been investigated in
detail w.r.t. algDNC analyses presented in this article. We provide
results about the amount of operations required for an analysis of
tandems as well as sink trees.

\subsection{Tandem of Length $h$}

We start with an analysis of the two algebraic analyses presented
in this article, SFA and algDNC, applied to a tandem of length $h$,
crossed by $m$ flows.

\subsubsection{SFA}

\label{subsec:SFA-tandem}The presented SFA tandem analysis requires
cross-traffic arrival bounds at every server on the tandem in order
to compute the foi's left-over service curve (step~1 in Section~\ref{subsec:Feed-forward-Analyses}).
Then, the end-to-end left-over service curve can be computed in order
to derive the foi's delay bound (step~2 in Section~\ref{subsec:Feed-forward-Analyses}).
For upper bounding the required amount of operations, we assume the
foi crosses the maximum amount of servers, giving it a path of length
$h$.

\begin{enumerate}

\item Bounding the arrivals of all cross-flows at every server in the tandem.

\end{enumerate}The literature~\cite{BouillardHDR} proposes to segregately
bound each cross-flow's arrival with a SFA. As the SFA only implements
the PBOO principle, this procedure became known as segregated PBOO
arrival bounding (segrPBOO)~\cite{BS15-2}. It starts at any server
on the foi's path and recursively backtracks cross-flows to derive
a bound on the arrival of each cross-flow there. We denote the distance
to the analyzed tandem's last server with $d$, $d\in\left\{ 0,\ldots,h-1\right\} $.
\begin{itemize}
\item Assume a server crossed by the foi at distance $d$. At this server,
there are $m-1<m$ cross-flows to be bounded, i.e., as many segrPBOO
arrival bounding recursions are started.
\item In this worst-case tandem, we can compute the number of recursions
invoked for distance $d+i+1$ by a server at distance $d+i$, $1\leq i\leq h-d-1$.
A server will invoke one segrPBOO arrival bounding for each of its
$m-1<m$ flows' $m-2<m$ cross-flows. Therefore, we can upper bound
the amount of these segrPBOO invocations with $m^{2}$.
\end{itemize}
Next, we need to count the operations that are actually invoked at
a distance $d+i$. Each invocation corresponds to the need to compute
an output bound for a single flow. Thus, we need to aggregate this
flow's $m-1$ cross-flows with $m-2$ operations, derive the flow's
left-over service curve and compute the actual output bound. This
results in $m$ operations at each server. 

For cross-traffic arrival bounding required by a server at distance
$d$ on the foi's path, we need to sum up the operations on the entire
part of the tandem to be backtracked by segrPBOO arrival bounding,
i.e., the servers at distances $i$, $d<i\leq h-1$. We saw that these
servers will be invoked multiple times, depending on the amount of
server in distances $j$, $d<j\leq i-1$, i.e., $\prod_{j=d+1}^{i-1}\left(m^{2}\right)^{j}$.
The $m$ operations for each flow at each server are scaled by this
number. Last, the server crossed by the foi at distance $d$ invokes
$m$ segrPBOO boundings itself. This gives the total number of operations:
\begin{eqnarray*}
m\cdot\sum_{i=d+1}^{h-1}m\cdot\prod_{j=d+1}^{i-1}\left(m^{2}\right)^{j} & = & m^{2-d\left(d+1\right)}\cdot\sum_{i=d+1}^{h-1}m^{\left(i-1\right)i}
\end{eqnarray*}
Last, this procedure needs to be repeated for all servers on the
foi's path. The total amount of operations to bound cross-traffic
arrivals is thus upper bounded by
\begin{equation}
\sum_{d=0}^{h-1}m^{2-d\left(d+1\right)}\cdot\sum_{i=d+1}^{h-1}m^{\left(i-1\right)i}\label{eq:SFAsegrPBOOtandem}
\end{equation}

\begin{enumerate}[resume]

\item Executing a SFA for the flow of interest.

\end{enumerate}The second analysis step corresponds to the compFFA
step~2, the tandem analysis on the foi's path. It consists of the
basic algDNC operations, cross-flow aggregation, left-over service
curve derivation, convolution to an end-to-end service curve, and
the eventual bounding of the foi's delay. These operations have to
be executed at each of the $h$ servers crossed by the foi. Remember
that there are $m$ flows at each server, one of which is the foi.

\emph{Cross-flow Aggregation:} After step~1, we know all $m-1$ segregately
derived cross-flow arrivals to be aggregated with $m-2$ operations.
I.e., there are $h\cdot(m-2)$ aggregations in total.

\vspace{1mm}

\emph{Left-over Service Curve Derivation:} The aggregated cross-traffic
arrivals are used to compute the left-over service curve. There are
$h$ left-over operations in total.

\vspace{1mm}

\emph{End-to-end Service Curve Computation:} Convolving the $h$ left-over
service curves requires $h-1$ operations.

\emph{Delay Bounding:} This is done with a single operation.

\vspace{1mm}

\emph{Total Amount of Operations for foi Analysis: 
\begin{eqnarray}
h\cdot(m-2)+h+(h-1)+1 & = & h\cdot m.\label{eq:SFAtandemDelay}
\end{eqnarray}
}

\vspace{1mm}

\emph{Total Amount of Operations for the tandem SFA (Equations (\ref{eq:SFAsegrPBOOtandem})
+ (\ref{eq:SFAtandemDelay})):
\[
h\cdot m\cdot\sum_{d=0}^{h-1}m^{2-d\left(d+1\right)}\cdot\sum_{i=d+1}^{h-1}m^{\left(i-1\right)i}
\]
}

which defines the an upper bound on the complexity of the SFA in a
tandem of length $h$, depending on the complexity imposed on the
operations by the shape of curves chosen to model flow arrivals and
service curves.

\subsubsection{The novel, exhaustive algDNC analysis}

\label{subsec:AlgDNC-tandem}Next, we evaluate the amount of operations
required in our novel algDNC analysis. For comparison with the SFA
above, we use the same tandem of $h$ servers entirely crossed by
$m$ flows.

\begin{enumerate}

\item Bounding the arrivals of all cross-flows at every server in the tandem.

\end{enumerate}In contrast to the SFA, we employ aggregation during
arrival bounding to mitigate the most basic cause of PSOO violations.
The corresponding procedure backtracks along paths of cross-flow aggregates.
It is known as aggrAB. We will combine aggrAB with the exhaustive
decomposition and the efficiency improvements of our novel algDNC
analysis.

In a tandem of $h$ servers that is entirely crossed by $m$ flows,
aggrAB only separates the foi and its single cross-traffic aggregate
of $m-1$ flows. I.e., the cross-traffic arrival bounding will not
become recursive like the SFA's segrPBOO. The cross-traffic aggregate
has a path of $h$ servers, i.e, it contains $h-1$ links which results
in $2^{h-1}$ different decompositions into sub-tandem sequences.
This leads to an average of $\frac{h+1}{2}$ tandems per decomposition,
as Proposition~\ref{prop:average-number-tandem-matching-1} shows.
Each tandem requires one output bound operation. Left-over service
curve operations are not required as there is not recursive bounding
of cross-traffic. However, tandems need to be convolved for the derivation
of output bounds arriving further to the end of the tandem. That results
in the need for $\frac{h+1}{2}-1=\frac{h-1}{2}$ convolution operations
on average. In total, these steps give us
\begin{eqnarray}
2^{h-1}\left(\frac{h+1}{2}+\frac{h-1}{2}\right) & = & h\cdot2^{h-1}\label{eq:algDNC-tandem-aggrAB}
\end{eqnarray}
algDNC operations to bound the foi's cross-traffic. 
\begin{prop}
\label{prop:average-number-tandem-matching-1}In a tandem of length
$h$, a decomposition's average amount of sub-tandem is $\frac{h+1}{2}.$
\end{prop}

\begin{proof}
There are $\binom{h-1}{k}$ possibilities to find $k-1$ subtandems
in a tandem of length $h$: only one ($\binom{h-1}{0}$) configuration
to have one subtandem, $h-1=\binom{h-1}{1}$ to have 2 subtandem etc.
Since the total numbers of tandem matchings equals $2^{h-1}$, we
obtain for the average number of subtandems
\[
\frac{\sum_{k=0}^{h-1}\binom{h-1}{k}\left(k+1\right)}{2^{h-1}}\,=\,\frac{\sum_{k=0}^{h-1}\binom{h-1}{k}+\sum_{k=0}^{h-1}k\binom{h-1}{k}}{2^{h-1}}\,=\,\frac{2^{h-1}+\left(h-1\right)\cdot2^{h-2}}{2^{h-1}}\,=\,\frac{\left(h+1\right)2^{h-2}}{2^{h-1}}\,=\,\frac{h+1}{2}.
\]
\end{proof}
\begin{enumerate}[resume]

\item Executing the novel algDNC analysis for the flow of interest.

\end{enumerate}

On the foi's path, we proceed similarly. The tandem consists of $h$
servers and can be decomposed into $2^{h-1}$ different sub-tandem
sequences. In contrast to arrival bounding, we need to derive one
left-over service curve per sub-tandem, yet, no output bound. This
gives us $\frac{h+1}{2}$ as in aggrAB. Additionally, we convolve
$\frac{h+1}{2}-1=\frac{h-1}{2}$ sub-tandems per decomposition on
average, again. Finally, the foi analysis requires one delay bounding.
Thus, the total amount of operations is
\begin{eqnarray}
2^{h-1}\left(\frac{h+1}{2}+\frac{h-1}{2}\right)+1 & = & h\cdot2^{h-1}+1.\label{eq:algDNC-tandem-foi}
\end{eqnarray}

\vspace{1mm}

\emph{Total Amount of Operations for the exhaustive tandem algDNC
analysis (Equations (\ref{eq:algDNC-tandem-aggrAB}) + (\ref{eq:algDNC-tandem-foi})):
}
\begin{eqnarray*}
h\cdot2^{h-1}+h\cdot2^{h-1}+1 & = & h\cdot2^{h}+1
\end{eqnarray*}

which defines the complexity of the exhaustive algDNC in a tandem
of $h$ servers that is entirely crossed by $m$ flows. Note, that
the result is independent of the actual size of $m$ as the aggrAB
convolves all cross-flows into a single cross-traffic aggregate. This
constitutes a major advantage over the segrPBOO arrival bounding of
the literature's SFA~\cite{BouillardHDR}.

\subsection{Full $k$-ary Sink-Tree Networks of Height $h$}

We continue with an analysis of SFA's and algDNC's application to
full $k$-ary sink-tree networks of a maximum height $h$. These sink
trees were already used to illustrate the combinatorial explosion
of optDNC's LP analysis (Section~\ref{subsec:LPcombinatorialExplosion}).

Such a sink tree has $n=\frac{k^{h+1}-1}{k-1}$ nodes, each node corresponds
to one server, and we assume the foi to originate at a leaf node in
order to cross the maximum amount of nodes, $h+1$ (the sink is at
height $0$).

\subsubsection{SFA}

For the SFA, the amount of operations in a full $k$-ary sink tree
of height $h$ is composed of the following parts:

\begin{enumerate}

\item Bounding the arrivals of all cross-flows at every node in the tree.

\end{enumerate}The SFA applies segrPBOO arrival bounding. I.e., we
need to recursively unfold the computations for each individual flow
at every node in the sink tree. We use structural information of the
sink tree to count the operations necessary to bound cross-traffic
arrivals for the foi:
\begin{itemize}
\item The amount of nodes in the full $k$-ary sink tree of height $h$
is $\frac{k^{h+1}-1}{k-1}<\frac{k}{k-1}k^{h}$. We assume that one
flow originates at every node and all flows cross the sink. Then,
there are $\frac{k}{k-1}k^{h}$ flows at the sink.
\item Each node at distance $d\in\left\{ 0,\ldots,h\right\} $ from the
sink is itself the root of a sub-tree. This sub-tree has $\frac{k^{h-d+1}-1}{k-1}<\frac{k}{k-1}k^{h-d}$
nodes, i.e., the it is crossed by $\frac{k}{k-1}k^{h-d}$ flows.
\item There are $k^{d}$ nodes at a distance of $d$ from the sink.
\end{itemize}
The segrPBOO arrival bounding starts at any node on the foi's path
and recursively derives a bound on the arrival of each cross-flow
there.
\begin{itemize}
\item Assume a node crossed by the foi at distance $d$. There are at most
$\frac{k}{k-1}k^{h-d}-1<\frac{k}{k-1}k^{h-d}$ cross-flows to be bounded,
i.e., as many segrPBOO recursions are started. They result in $\frac{k}{k-1}k^{h-d}$
recursions invoking further recursions on the next level in the sink
tree, i.e., distance $d+1$ from the sink.
\item In the sink tree, we can easily compute the number of times that the
$k^{d+i}$ nodes at distance $d+i$ will invoke operations at nodes
at distance $d+i+1$ from the sink, $1\leq i\leq h-d$. The $k^{d+i}$
nodes will each invoke one segrPBOO arrival bounding for each of their
$\frac{k}{k-1}k^{h-(d+i)}$ flows' $\frac{k}{k-1}k^{h-(d+i)}-1<\frac{k}{k-1}k^{h-(d+i)}$
cross-flows. Note that, in the sink tree, it is guaranteed that all
flows will need to bound their cross-traffic as the arrival bound
at the next node closer to the sink is required. Also note, that we
assume all cross-flows always arrive from a node one level further
away from the source. Therefore, we can upper bound the amount of
invocations at level $d+i+1$ as we neglect that flows originate at
distance $d+i$; the flows for which the arrival curve is already
known. Thus, level $d+i$ will invoke $k^{d+i}\cdot\left(\frac{k}{k-1}k^{h-(d+i)}\right)^{2}$
segrPBOO arrival boundings at level $d+i+1$.
\end{itemize}
Next, we need to count the operations that are actually invoked at
a level $d+i$. For non-leaf nodes, these are $\frac{k}{k-1}k^{h-(d+i)}-2$
aggregations of cross-flow arrival bounds, $1$ left-over service
curve derivation and $1$ output bound. I.e., $\frac{k}{k-1}k^{h-(d+i)}$
operations in total. At leaf-nodes, we only have a single output bound
operation each, as there in only a single flow present at this level.

For cross-traffic arrival bounding of a server at distance $d$ on
the foi's path, we need to sum up the operations in the entire sink
tree backtracked by segrPBOO arrival bounding, i.e., the nodes at
distances $i$, $d<i\leq h$. We saw that these $k^{i}$ nodes will
be invoked depending on the size of the sink tree between $d$ and
$i$, i.e., distances $j$, $d<j<i-1$, causing $\prod_{j=1}^{i-1}k^{d+j}\cdot\left(\frac{k}{k-1}k^{h-(d+j)}\right)^{2}$
invocations. These, in turn, cause $\frac{k}{k-1}k^{h-(d+i)}$ operations
at non-leaf nodes and a single one at leaf nodes. Last, the node crossed
by the foi at distance $d$ invokes $\frac{k}{k-1}k^{h-d}$ segrPBOO
boundings itself. This gives the total number of operations in segrPBOO
arrival bounding for a node at distance $d$ from the sink:
\begin{align*}
 & \frac{k}{k-1}k^{h-d}\cdot\left(\sum_{i=d+1}^{h}\left(\frac{k}{k-1}k^{h-\left(d+i\right)}\cdot\prod_{j=1}^{i-1}k^{d+j}\left(\frac{k}{k-1}k^{h-\left(d+j\right)}\right)^{2}\right)+\prod_{j=d}^{h+1}k^{h-j}\left(\frac{k}{k-1}k^{h-\left(d+j\right)}\right)^{2}\right)\\
= & \left(\frac{k}{k-1}\right)^{3}k^{h-d}\cdot\left(\frac{k}{k-1}k^{3h-2d}\sum_{i=d+1}^{h}\left(k^{-\frac{1}{2}i\left(i-3\right)}\right)+k^{-\frac{3}{2}\left(h^{2}+6h+2-d^{2}-d\right)}\right)\\
\leq & \,8k^{4h-3d}\cdot\left(2\sum_{i=d+1}^{h}\left(k^{-\frac{1}{2}i\left(i-3\right)}\right)+k^{\frac{3}{2}\left(-h^{2}-8h+d^{2}+3d-2\right)}\right),
\end{align*}
where we have used that $k\geq2$, the outdegree of a node defining
the difference between tandem and tree networks.

Finally, this procedure needs to be repeated for all nodes on the
foi's path. For upper bounding the required amount of operations,
we assume the foi to cross the maximum amount of nodes, giving it
a path of length $h+1$, i.e., $d\in\left\{ 0,\ldots,h\right\} $.
The total amount of operations to bound cross-traffic arrivals is
thus upper bounded by
\[
8k^{4h}\sum_{d=0}^{h}k^{-3d}\cdot\left(2\sum_{i=d+1}^{h}\left(k^{-\frac{1}{2}i\left(i-3\right)}\right)+k^{\frac{3}{2}\left(-h^{2}-8h+d^{2}+3d-2\right)}\right).
\]

\begin{enumerate}[resume]

\item Executing a SFA for the flow of interest.

\end{enumerate}The second analysis step corresponds to the compFFA
step~2. It equals the same step in the tandem, yet for $h+1$ nodes.\vspace{3mm}

\emph{Total Amount of Operations for the Sink-tree SFA
\[
\left(h+1\right)\cdot m\cdot8k^{4h}\sum_{d=0}^{h}k^{-3d}\cdot\left(2\sum_{i=d+1}^{h}\left(k^{-\frac{1}{2}i\left(i-3\right)}\right)+k^{\frac{3}{2}\left(-h^{2}-8h+d^{2}+3d-2\right)}\right).
\]
}which defines the complexity of the SFA in a full $k$-ary sink
tree of maximum height $h$.

\subsubsection{The novel, exhaustive algDNC analysis}

\begin{enumerate}

\item Bounding the arrivals of all cross-flows at every node in the sink tree.

\end{enumerate}\textbf{}This analysis step covers compFFA step~1
presented in Section~\ref{subsec:Feed-forward-Analyses} and consists
of the algDNC operations output bounding and aggregation of flows.
For aggrAB with efficiency improvements, we can derive the amount
of operations in this recursion as follows.

\emph{Output Bounding:} Derivation of these bounds requires one (min,+)-deconvolution
at all nodes except the sink whose output is not considered in any
analysis. Neither do we consider the the foi's source node. Therefore,
we require $n-2$ deconvolution operations.

\emph{Aggregation:} All non-leaf nodes have $k$ inlinks, each contributing
one flow (aggregate), all of which need to be aggregated during the
arrival bounding in the sink tree. A full $k$-ary sink tree of height
$h$ has $k^{h}$ leaf nodes and thus $n-k^{h}$ non-leaf nodes. Note,
that the leaf nodes' share of the total network is increasing in $k$.
The entire sink tree has $\left(n-k^{h}\right)\cdot k=n-1$ links
whose flows need to be aggregated. One exception exists: similar to
output bounding, the link after the foi's source does not contribute
to the aggregation requirement as it is only crossed by the foi. Therefore,
we obtain a total of $n-2$ algDNC aggregations operations.

\vspace{1mm}

\emph{Total Amount of Operations for Arrival Bounding:
\begin{equation}
n-2+n-2=2n-4.\label{eq:k-ary-arrive-bound-SFA}
\end{equation}
}

\begin{enumerate}[resume]

\item Executing the novel algDNC analysis for the flow of interest.

\end{enumerate}The foi's path is a tandem where we know all the required
cross-traffic arrival bounds for our algDNC analysis. I.e., we take
the same steps as shown in Section~\ref{subsec:AlgDNC-tandem}'s
foi analysis. Note, however, that this tandem has an additional hop
and consists of consisting of $h+1$ nodes. Then, we obtain $2^{h}$
different decompositions into sub-tandem sequences with an average
of $\frac{h+2}{2}$ sub-tandems. This results in 
\begin{equation}
2^{h}\left(\frac{h+2}{2}+\frac{h}{2}\right)+1=2^{h}\left(h+1\right)+1<4n\cdot\log_{2}\left(n\right)+13\label{eq:k-ary-foi-ana-algDNC}
\end{equation}
algDNC operations for bounding the foi delay.

\vspace{3mm}
\emph{Total Amount of Operations for the exhaustive Sink-tree algDNC
analysis (Equations (\ref{eq:k-ary-arrive-bound-SFA}) + (\ref{eq:k-ary-foi-ana-algDNC})):
\[
2n\cdot\left(1+2\cdot\log_{2}\left(n\right)\right)+9,
\]
}which defines the complexity of the exhaustive algDNC in a full $k$-ary
sink tree of maximum height $h$, depending on the complexity imposed
on the operations by the shape of curves chosen to model flow arrivals
and nodal service.

\section{Cost Reduction by Parallelized Optimization}

\label{sec:EvalCostParallelCPLEX}\begin{figure*}[!t]
\begin{centering}
\includegraphics[width=0.8\textwidth]{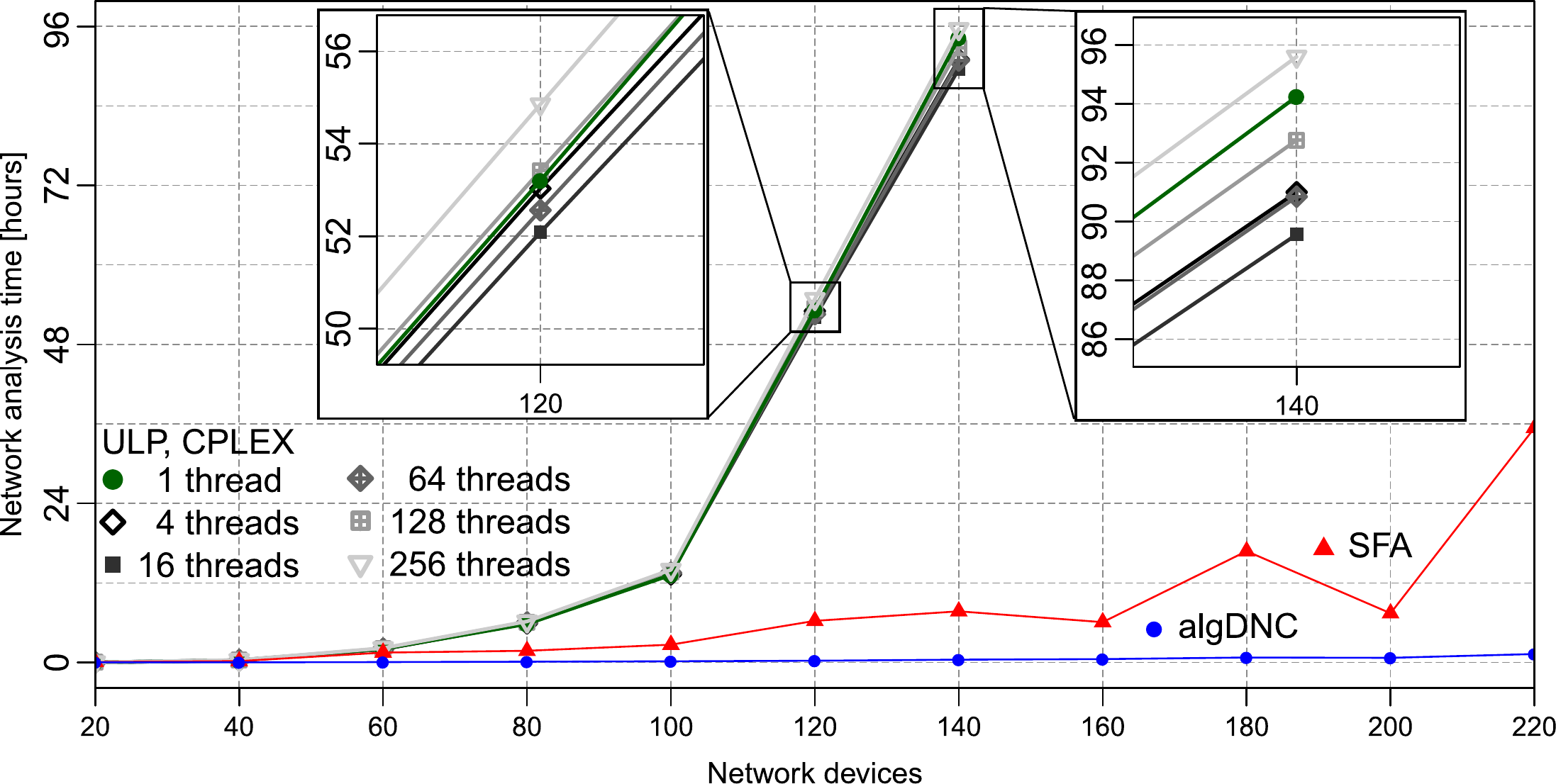}
\par\end{centering}
\vspace{3mm}
\begin{centering}
%{\smaller
\begin{tabular}[b]{|r||c|c|c|}
\hline
Net & \multicolumn{2}{|c|}{ULP, CPLEX} & algDNC\tabularnewline
    & \scriptsize{(256 threads)} & \scriptsize{(16 threads)} & \scriptsize{(1 thread)}\tabularnewline
\hline 
\hline
20 & 0:01:04 & 0:00:49 & 0:00:51\tabularnewline
\hline
40 & 0:21:54 & 0:20:17 & 0:01:12\tabularnewline
\hline
60 & 2:10:14 & 2:00:52 & 0:03:21\tabularnewline
\hline
80 & 6:22:44 & 5:57:59 & 0:06:44\tabularnewline
\hline
100 & 13:59:08 & 13:23:08 & 0:10:39\tabularnewline
\hline
120 & 54:51:04 & 52:05:02 & 0:15:45\tabularnewline
\hline
140 & 95:36:08 & 89:33:38 & 0:25:09\tabularnewline
\hline
160 & -- & -- & 0:30:18\tabularnewline
\hline
180 & -- & -- & 0:44:46\tabularnewline
\hline
200 & -- & -- & 0:41:56\tabularnewline
\hline
220 & -- & -- & 1:16:23\tabularnewline
\hline
240 & -- & -- & 1:39:18\tabularnewline
\hline
\multicolumn{4}{c}{\footnotesize Network analysis time}\tabularnewline
\multicolumn{4}{c}{\footnotesize (hh:mm:ss)}\tabularnewline
\end{tabular}
%}
\par\end{centering}
\caption{\label{fig:AnalysisEffort-Parallel}Network analysis cost on a modern
many-core architecture allowing for massively parallelized optimization
over slower individual CPU cores. The slowdown for single-threaded
execution increases algDNC network analysis times, yet, parallelized
optimization does not allow optDNC to attain a similar performance
in any network larger than $20$ devices.}
\end{figure*}

Figure~\ref{fig:AnalysisEffort-Parallel}
depicts the results of a second run of the same sample analyses of
Section~\ref{fig:Eval-AFDX-DelayBounds} and~\ref{fig:Eval-AFDX-RunTimes}.
They were executed on a Colfax Ninja Workstation equipped with an
Intel Xeon Phi 7210 CPU. In contrast to the Intel Xeon server CPU
of the previous evaluation, the Xeon Phi is a modern many-core CPU
offering 64~physical cores that can, employing simultaneous multithreading,
process up to 256~threads in parallel. These cores do, however, clocked
at only 1.3GHz. The workstation offers 110GB RAM to cope with the
massively parallelized optimization's memory demand without swapping
to disk.

The decreased single-thread performance negatively impacts the network
analysis times of both algebraic DNC analyses. While only the increase
in SFA analysis time is already visible in Figure~\ref{fig:AnalysisEffort-Parallel},
on the left, the measured algDNC times given in the table on the right
show a significant increase as well. Relative to the server-grade
CPU, SFA network analysis takes between $2.0$ and $3.82$ times as
long, with an average of $3.27$. For the algDNC analysis, the slowdown
is between factor $4.05$ and $4.51$, with an average of $4.19$.
Similarly, the ULP analyses became slower. For both, the 4-threaded
and the single-threaded analysis, this slowdown was between factor
$1.16$ and factor $3.33$. Although this slowdown is smaller than
algDNC's one, parallelized optimization did not allow optDNC to attain
a similar performance in any network larger than $20$ devices.

Most interestingly, this disadvantage of modern many-core CPUs could
not be compensated with massively parallelizing the ULP's optimization.
In fact, network analyses using all $256$~logical cores were the
slowest among the tested parallelization levels $1$, $4$, $16$,
$64$, $128$, and $256$. Fastest analysis times could be achieved
with a parallelization over a maximum of $16$~threads, far below
the amount of physical CPU cores (see Figure~\ref{fig:AnalysisEffort-Parallel},
inlets zooming into results for $120$ and $140$ devices). The difference
between the best and worst setting for parallelization is $\approx\!23\%$
in the smallest network but it decreases to constantly less than $8\%$
in all larger networks. Thus, massively parallelized optimization
did not yield considerable improvements for current optDNC analysis.

\end{document}